\documentclass[11pt,letterpaper, onecolumn]{quantumarticle}
\pdfoutput=1
\usepackage{color,fullpage,url,booktabs}
\usepackage[numbers,sort&compress]{natbib}
\usepackage{amsmath, amsthm, amssymb}  %cite 
\usepackage{colorprofiles}
\usepackage[a-1b,mathxmp]{pdfx}[2018/12/22]
\hypersetup{pdfstartview=}

\makeatletter
\AtBeginDocument{\let\mathaccentV\AMS@mathaccentV}
\makeatother

\usepackage{listings}
\usepackage{graphicx}
\usepackage{tikz}
\usepackage{float}

\newcommand{\nc}{\newcommand}
\nc{\rnc}{\renewcommand}

\newcommand{\ket}[1]{\left|#1\right\rangle}

\newcommand{\vev}[1]{\left\langle #1\right\rangle}

\DeclareMathOperator{\diag}{diag}

\DeclareMathOperator{\poly}{poly}

\DeclareMathOperator{\tr}{tr}

\def\be#1\ee{\begin{equation}#1\end{equation}}
\def\bea#1\eea{\begin{eqnarray}#1\end{eqnarray}}
\def\beas#1\eeas{\begin{eqnarray*}#1\end{eqnarray*}}
\def\ba#1\ea{\begin{align}#1\end{align}}
\def\bas#1\eas{\begin{align*}#1\end{align*}}
\def\bpm#1\epm{\begin{pmatrix}#1\end{pmatrix}}

\def\eq#1{\eqref{eq:#1}}

\def\L{\left} 
\def\R{\right}
\def\ra{\rightarrow}

\nc{\grad}{{\vec{\nabla}}}

\newtheorem{thm}{Theorem}
\newtheorem*{thm*}{Theorem}

\newtheorem{prop}[thm]{Proposition}
\newtheorem{proto}{Protocol}

\theoremstyle{definition}

\newtheorem{dfn}[thm]{Definition}
\theoremstyle{plain}
\newtheorem{theo}[thm]{Theorem}

\newtheorem{lemma}[thm]{Lemma}

\makeatletter
\newtheorem*{rep@theorem}{\rep@title}
\newcommand{\newreptheorem}[2]{%
\newenvironment{rep#1}[1]{%
 \def\rep@title{#2 \ref{##1} (restatement)}%
 \begin{rep@theorem}}%
 {\end{rep@theorem}}}
\makeatother

\newreptheorem{thm}{Theorem}
\newreptheorem{lem}{Lemma}

\def\cB{\mathcal{B}}

\def\cH{\mathcal{H}}

\def\cO{{\cal O}}
\def\cP{\mathcal{P}}
\def\cQ{\mathcal{Q}}

\def\cZ{\mathcal{Z}}

\def\bv{\mathbf{v}}
\def\bx{\mathbf{x}}
\def\by{\mathbf{y}}
\def\bz{\mathbf{z}}
\def\bfeta{\boldsymbol{\eta}}

\DeclareMathOperator*{\bbE}{\mathbb{E}}
\DeclareMathOperator*{\bbP}{\mathbb{P}}

\def\benum{\begin{enumerate}}
\def\eenum{\end{enumerate}}
\def\bit{\begin{itemize}}
\def\eit{\end{itemize}}
\newcommand{\fig}[1]{Fig.~\ref{fig:#1}}
\newcommand{\secref}[1]{Section~\ref{sec:#1}}

\newcommand{\lemref}[1]{Lemma~\ref{lem:#1}}

\nc{\todo}[1]{\textcolor{red}{todo: #1}}

\def\begsub#1#2\endsub{\begin{subequations}\label{eq:#1}\begin{align}#2\end{align}\end{subequations}}
\nc\qand{\qquad\text{and}\qquad}
\nc\mnb[1]{\medskip\noindent{\bf #1}}

\nc{\pder}[2]{\frac{\partial {#1}}{\partial {#2}}}
\nc{\p}{\partial}

\begin{document}

\title{Rapid mixing of path integral Monte Carlo for 1D stoquastic Hamiltonians}
%\title{A polynomial-time algorithm for partition functions and constant gap adiabatic evolutions of 1-dimensional stoquastic Hamiltonians}
\author{Elizabeth Crosson}
\affiliation{Center for Quantum Information and Control, University of New Mexico.}
\email{crosson@unm.edu}
\author{Aram W. Harrow }
\affiliation{Center for Theoretical Physics, Massachusetts Institute of Technology.}
\email{aram@mit.edu}

\maketitle
\vspace{-20pt}
\begin{abstract}
Path integral quantum Monte Carlo (PIMC) is a method for estimating thermal equilibrium properties of stoquastic quantum spin systems by sampling from a classical Gibbs distribution using Markov chain Monte Carlo.  The PIMC method has been widely used to study the physics of materials and for simulated quantum annealing, but these successful applications are rarely accompanied by formal proofs that the Markov chains underlying PIMC rapidly converge to the desired equilibrium distribution.  

In this work we analyze the mixing time of PIMC for 1D stoquastic Hamiltonians, including disordered transverse Ising models (TIM) with long-range algebraically decaying interactions as well as disordered XY spin chains with nearest-neighbor interactions.    By bounding the convergence time to the equilibrium distribution we rigorously justify the use of PIMC to approximate partition functions and expectations of observables for these models at inverse temperatures that scale at most logarithmically with the number of qubits. 

The mixing time analysis is based on the canonical paths method applied to the single-site Metropolis Markov chain for the Gibbs distribution of 2D classical spin models with couplings related to the interactions in the quantum Hamiltonian. Since the system has strongly nonisotropic couplings that grow with system size, it does not fall into the known cases where 2D classical spin models are known to mix rapidly.
\end{abstract}

\section{Introduction}

The application of Markov chain Monte Carlo methods to thermal states of a limited class of quantum systems was first proposed by Suzuki, Miyashita and Kuroda 40 years ago \cite{suzuki-1977}.   Their method applied to Hamiltonians with real and nonpositive off-diagonal matrix elements in a particular basis, a property that was called ``being free of the sign problem''  for many years until the term ``stoquastic'' was eventually adopted \cite{bravyi-2006} to emphasize a connection with nonnegative matrices.  Since then there have been a growing number of works on rigorous classical simulations for various stoquastic systems~\cite{bravyi-2008,bravyi-2014,crosson2016simulated,bravyi2016polynomial}.    Many quantum systems of physical and computational interest are stoquastic, including spinless particles moving on arbitrary graphs with position dependent interactions, as well as generalized transverse Ising models, which are particularly notable for their use in quantum annealing~\cite{kadowaki1998quantum, farhi-2000,albash2016adiabatic}.  

\paragraph{Statement of results.}A general $n$-qubit 1D stoquastic Hamiltonian with nearest-neighbor interactions has the form,
\begin{equation}
\mathcal{H} = \sum_{j=1}^{n} H_{j,j+1}, \label{eq:Ham}
\end{equation}
where each $H_{j,j+1}$ acts non-trivially on (at most) the sites $j,j+1$ (with $n+1$ identified with 1) and  each $\|H_{j,j+1}\| \leq 1$.  Define the ``computational basis'' $\{\ket{1},\ket{-1}\}$ for each qubit to be the eigenstates of the $\sigma^z$ operator. 
We assume $\cH$ is stoquastic in the computational basis~\cite{bravyi-2006, bravyi-2008}, which means that in this basis each $H_{j,j+1}$ is a real symmetric matrix with nonpositive off-diagonal entries,
\be
\langle z | H_{j,j+1} |z'\rangle \leq 0 \; \textrm{for all}\; z,z'\in \{-1,1\}^n \; \textrm{with}\; z\neq z'.
\ee
%where $|\pm 1\rangle$ are the states that satisfy $\sigma^z |\pm 1\rangle = \pm |\pm 1\rangle$. 
 For a given inverse temperature $\beta$ we consider the quantum partition function $\mathcal{Z}_\beta$ and the quantum Gibbs state $\rho_\beta$,
\begin{equation}
\rho_\beta := \frac{e^{-\beta \mathcal{H}}}{\cZ_\beta} \; \; , \; \; \cZ_\beta  := \tr \left(e^{-\beta \mathcal{H}}\right), \label{eq:Part}
\end{equation} 
%The property of $H$ being stoquastic implies that all the matrix elements of $\rho$ are real and non-negative in the computational basis.   %Stoquastic Hamiltonians are also said to be ``sign-problem-free.'' 
and also the expectation value of observables $\mathcal{Q}$ in the thermal state, $\langle \mathcal{Q}\rangle_\beta := \tr [\cQ \rho_\beta]$.  
For the special class of systems called 1D generalized transverse Ising models (TIM) we go beyond \eqref{eq:Ham} and allow systems with long-range interactions that decay sufficiently quickly (faster than inverse square) with distance,  %with transverse field strengths bounded below by $\Gamma > 0$ that does not go to zero with the system size,
\begin{align}
&\mathcal H = - \sum_{j=1}^n \Gamma_j\sigma^x_j +\sum_{j=1}^n K^z_j \sigma^z_j + \sum_{j,k=1}^n K^{zz}_{jk} \sigma^z_{j}\sigma^z_{k} \quad , \quad \Gamma:= \min_{j = 1,...,n} \Gamma_j > 0, \\  &K^{z}_j \in [-1,1]  \quad ,  \quad \textrm{ and } |K^{zz}_{jk} | \leq |i - j|^{-(2+ \xi)} \textrm{ for } \xi > 0.  \label{eq:tim} 
\end{align}
%\Gamma:= \min_{j = 1,...,n} \Gamma_j > 0.
The form \eq{tim} simplifies both the presentation of the PIMC method and our proof that the mixing time of PIMC for these models is at most $\textrm{poly}(n,e^{\beta}, \Gamma^{-1})$.  We also show the mixing time of PIMC is at most $\poly(n,e^{\beta}, \Gamma^{-1})$ for generalized TIM with $\sigma^x \sigma^x$ and $\sigma^y \sigma^y$ interactions,
\ba
H =- \sum_{j=1}^n \Gamma_j\sigma^x_j -\sum_{j=1}^n K_j^{xx}\sigma^{x}_j\sigma^x_{j+1}  -\sum_{j=1}^n K_j^{yy}\sigma^{y}_j\sigma^y_{j+1} +  \sum_{j=1}^n \left (K^{zz}_j \sigma^z_j \sigma^z_{j+1} + K^z_j \sigma^z_j   \right),  \label{eq:xxham}\\
K^{zz}_j, K^{z}_j \in [-1,1] \; \;, \; \; K^{xx}_j \geq 0 \; \; , \; \; K^{yy}_j \in [-K^{xx}_j,K^{xx}_j] \; \;  \textrm{and} \; \; \Gamma:= \min_{j = 1,...,n} \Gamma_j > 0.\nonumber
\ea
The mixing time bounds can be used to construct an FPRAS (fully polynomial-time randomized approximation scheme) for the partition function of models of the form \eq{tim} and \eq{xxham} and inverse temperatures $\beta =\cO(\log n)$.   In addition to the partition function there is a similar FPRAS for the expectation value of observables that are sparse and row-computable (in the computational basis) in the thermal state $\rho_\beta$ for $\beta = \cO(\log n)$.   For general 1D stoquastic Hamiltonians \eq{Ham} we prove that the PIMC mixing time is at most quasi-polynomial in the system size, which yields a quasi-polynomial time approximation algorithm for the partition function and for the approximation of observables.  %The assumption of a nonzero transverse field is imposed so that the Hamiltonians \eq{tim} and \eq{xxham} are irreducible matrices in the computational basis (see Appendix~\ref{sec:stoquasticBackground}).  The important feature of the models \eqref{eq:xxham} for our analysis is that the off-diagonal part of each $H_{j,j+1}$ commutes with $\sigma^x_j \sigma^x_{j+1}$.  %For these systems we obtain a bound on the PIMC mixing  FPRAS for the partition function and for the expectation value of observables in the thermal state, at inverse temperatures up to $\beta =\cO(\log n)$.
\begin{theo}\label{theo:tim}
There is an algorithm which takes as input an inverse temperature $\beta$ and an $n$-qubit 1D stoquastic Hamiltonian $H$ and outputs an estimate $\tilde Z_\beta$ of the quantum partition function $\cZ_\beta$.

If also given an observable $\cQ\succeq 0$ that is sparse and row-computable in the computational basis, the algorithm outputs an estimate $\tilde \cQ$ of its expectation value in the state $\rho_\beta$.

These estimates satisfy
\ba |\tilde \cZ_\beta - \cZ_\beta| & \leq \delta_{\text{mult}}\cZ_\beta + \delta_{\text{add}} \\
|\tilde \cQ - \tr[\cQ \rho_\beta]| & \leq \delta_{\text{mult}}\tr[\cQ\rho_\beta] + \delta_{\text{add}} .\ea
with probability $\geq 1- \delta_{\text{fail}}$.
The algorithm runs in time at most 
$$\poly(n,e^{\beta}, \Gamma^{-1}, \delta_{\text{mult}}^{-1},\log\left(\delta_{\text{add}}^{-1}\right), \log(\delta_{\text{fail}}^{-1}))$$
for models of the form \eq{tim} and \eq{xxham}, and in time at most 
$$\poly(2^{\beta\left(\log \frac{n}{\delta_{\text{mult}}}\right)^2},\log(1/\delta_{\text{add}}), \log(1/\delta_{\text{fail}}))$$
for models of the form \eq{Ham}.
\end{theo}

\paragraph{Relation to previous work.} A variety of quantum Monte Carlo methods have been devised for estimating equilibrium properties of stoquastic Hamiltonians, which can be broadly divided into diffusion methods~\cite{foulkes2001quantum, stella2007quantum,jarret2016adiabatic} and path integral methods~\cite{suzuki-1977,suzuki-1986, Martonak-2002,albash2017off,sandvik1999stochastic}.  Diffusion Monte Carlo methods use Hamiltonian matrix elements to define transition rates for a substochastic random process, along with birth and death population dynamics or importance sampling to control statistical fluctuations.  Examples of proven efficient convergence for diffusion Monte Carlo methods include simulating frustration-free stoquastic adiabatic computation~\cite{bravyi-2008} and simulating the ``go with the winners'' method~\cite{jarret2017substochastic}.  Disadvantages of diffusion methods include quantum amplitude vs quantum probability obstructions for some variants~\cite{jarret2016adiabatic,bringewatt2018diffusion}, as well as the need for importance sampling with trial wave functions that may in general not be available~\cite{bravyi-2014}.  

PIMC methods, in contrast, are based on Markov chains defined on an enlarged state space consisting of many coupled copies of a high-temperature classical system along a so-called ``imaginary-time'' direction~\cite{suzuki-1977}.  PIMC has been used to give an FPRAS for the ferromagnetic case of TIM~\cite{bravyi-2014} on general interaction graphs by a reduction to a classic result on sampling the Gibbs distribution of classical ferromagnetic Ising models \cite{jerrum-1993}.  Also, the Glauber dynamics for an infinite dimensional variant of PIMC has been shown to have optimal temporal mixing for ferromagnetic TIM at high temperature \cite{Cipriani-2010} and for ferromagnetic TIM on regular trees \cite{Martinelli-2011}.  Here ``ferromagnetic'' means of the form in \eq{tim} but with each $K_j^z=0$ and each $K_j^{zz} \leq 0$, i.e.~lower energy is achieved by aligning adjacent spins.  Since such interactions cannot be frustrated, these models are considered easier to analyze.
Disadvantages of PIMC methods include the significant overhead created by adding the imaginary-time direction, and also the possibility of slow mixing due to topological obstructions~\cite{andriyash2017can, hastings-2013}.  The general presence of obstructions to rapid mixing for the standard versions of both diffusion and path integral methods also motivates the continuing development of new algorithms for simulating stoquastic Hamiltonians~\cite{albash2017off, bravyi2016polynomial}.  

To our knowledge, the simulations in this work are the first provably efficient Monte Carlo methods for a large class of generalized TIM systems, which differ from ferromagnetic TIM by allowing couplings of both signs between neighboring spins, as well as local fields of both signs.  A generalized TIM in 2D or higher dimensions can encode NP-complete constraint satisfaction problems, which is the basis for the use of these models in quantum annealing~\cite{kadowaki1998quantum, farhi-2000,albash2016adiabatic}.   Besides their relative simplicity from the standpoint of experimental implementation, generalized TIM on interaction graphs of degree 3 are also known to be universal for stoquastic quantum annealing~\cite{bravyi-2014b,cubitt2017universal}.  The application of Monte Carlo methods to stoquastic quantum annealing is called simulated quantum annealing, and a significant open question is whether some variant of simulated quantum annealing can efficiently simulate stoquastic quantum annealing.  This open question has motivated several recent proposals for non-stoquastic quantum annealing~\cite{CrossonFLLS14,hormozi2017nonstoquastic,vinci2017non,susa2018exponential}.  Some recent progress has been made by analyzing specific models for which both quantum annealing and simulated quantum annealing can be exponentially faster than classical simulated annealing~\cite{farhi-2002, Reichardt-2004,  kong2015performance,Lidar-2015b, brady2016quantum, brady2016spectral,crosson2016simulated,jiang2017scaling}.

Another class of simulation methods for 1D quantum systems are those based on matrix product states~\cite{verstraete2008matrix}.  These methods do not require the Hamiltonian to be stoquastic and have been applied in rigorous classical simulations of systems with limited entanglement~\cite{hastings-2009,landau2013polynomial,arad2016rigorous}.   Their runtime scales exponentially with spectral gap, while our runtime bound for PIMC scales exponentially with the inverse temperature, which is qualitatively similar.
These kinds of simulations are not expected to obsolete PIMC in general, however, as the latter is used in practice to simulate highly entangled stoquastic systems.  Furthermore, the version of PIMC that we consider is closely related to the practical implementations of PIMC, and so this work can be seen as a rigorous analysis of a heuristic approximation algorithm that was originally developed and applied with great success by the physics community.

\paragraph{Organization of the remaining sections.}  \secref{overview} reviews the PIMC method and sketches the proof of our mixing time analysis.  \secref{qtc} defines the Suzuki-Trotter approximation to the partition function which is used to map 1D stoquastic Hamiltonians onto 2D classical spin systems.  \secref{approxObservables} and \secref{estimatePartitionFunction} describe the use of samples from the Gibbs distribution of the 2D classical system to approximate quantum observables and the quantum partition function.  \secref{restricting} proves a concentration theorem for PIMC that justifies the analysis of a restricted state space in later sections.  \secref{markovchains} reviews the canonical paths method that we use to analyze the PIMC mixing time.  This method is then applied to generalized TIM in \secref{transverseIsing}, generalized TIM with XX interactions in \secref{XX}, and general 1D stoquastic Hamiltonians in \secref{general1D}.
\section{Overview}\label{sec:overview}

In the PIMC method the Hamiltonian \eqref{eq:Ham} is related by Suzuki's quantum-to-classical mapping \cite{suzuki-1986,suzuki-2013}  to a system of $\{\pm 1\}$ classical spins on a 2D lattice $\Lambda = \{1,\ldots,L\}\times \{1,\ldots,n\}$.  The new site index $\{1,\ldots,L\}$ is known in physics as the ``imaginary-time'' direction, and it is often useful to think of the set of spin configurations $\Omega:= \{-1,1\}^{n\times L}$ as consisting of ``time slices'' along the $L$ direction, i.e. for $\mathbf{z} \in \Omega$ we write $\mathbf{z} := (z_1,\ldots,z_L)$ with each $z_i \in \{-1,1\}^n$.   The mapping introduces a local energy function $E_\beta:\Omega\rightarrow \mathbb{R}$ on the classical spin configurations in such a way that an estimate of the classical partition function $Z := \sum_{\mathbf{z} \in \Omega}e^{-\beta E_\beta(\mathbf{z})}$ can be used to estimate the quantum partition function $\mathcal{Z}$, and so that samples from the Gibbs distribution $\pi(\mathbf{z}) := e^{-\beta E_\beta(\mathbf{z})}/Z$ can be used to estimate expectation values in the quantum Gibbs state.  For the nearest-neighbor version of the transverse Ising models in \eqref{eq:tim} the energy function $E_\beta$ is
\begin{equation}
E_\beta(\mathbf{z}) = \sum_{(i,j) \in \Lambda} \left( \frac{K^{zz}_{j,j+1}}{L} z_{i,j} z_{i,j+1} + \frac{K^z_j}{L} z_{i,j} - \beta^{-1}J_i z_{i,j}z_{i+1,j}\right) \quad , \quad  z_{L+1,j} := z_{1,j} ,\label{eq:timenergy}
\end{equation}
where $z_{i,j}$ denotes the spin in configuration $\mathbf{z}$ at the site $(i,j) \in \Lambda$, and $J_i := \frac{1}{2} \log \coth \left(\frac{\beta \Gamma_i}{L}\right)$ (see \secref{qtc}).   The couplings along the imaginary-time direction are ferromagnetic, but the couplings along the spatial direction can have varying signs as well as local fields.  Our choices of parameters are such that $L = \poly(n)$, and so the general case considers $J_i$ with magnitude up to $\mathcal{O}(\log n)$.  Therefore the couplings of the 2D model are highly anisotropic, as in figure~\ref{fig:anisotropic}.  
\begin{figure}[H]
\captionsetup{width=0.8\textwidth}
\begin{center}\begin{tikzpicture}[xscale=1,yscale=.7]
%\draw  (-4,3.5) rectangle (5.5,-1.5);
\node [draw,shape=circle,inner sep=1.5pt,fill] (v14) at (-3,3) {};
\node [draw,shape=circle,inner sep=1.5pt,fill] (v15) at (-3,2) {};
\node [draw,shape=circle,inner sep=1.5pt,fill] (v17) at (-3,1) {};
\node [draw,shape=circle,inner sep=1.5pt,fill] (v19) at (-3,0) {};
\node [draw,shape=circle,inner sep=1.5pt,fill] (v21) at (-3,-1) {};
\node [draw,shape=circle,inner sep=1.5pt,fill] at (-1.5,-1) {};
\node [draw,shape=circle,inner sep=1.5pt,fill] at (-1.5,0) {};
\node [draw,shape=circle,inner sep=1.5pt,fill] at (-1.5,1) {};
\node [draw,shape=circle,inner sep=1.5pt,fill] at (-1.5,2) {};
\node [draw,shape=circle,inner sep=1.5pt,fill] at (-1.5,3) {};
\node at (1,-2) {\large space};
\node [rotate=90] at (-4,1) {\large imaginary-time};
\node [draw,shape=circle,inner sep=1.5pt,fill] at (0,-1) {};
\node [draw,shape=circle,inner sep=1.5pt,fill] at (0,0) {};
\node [draw,shape=circle,inner sep=1.5pt,fill] at (0,1) {};
\node [draw,shape=circle,inner sep=1.5pt,fill] at (0,2) {};
\node [draw,shape=circle,inner sep=1.5pt,fill] at (0,3) {};
\node [draw,shape=circle,inner sep=1.5pt,fill] at (1.5,3) {};
\node [draw,shape=circle,inner sep=1.5pt,fill] at (1.5,2) {};
\node [draw,shape=circle,inner sep=1.5pt,fill] at (1.5,1) {};
\node [draw,shape=circle,inner sep=1.5pt,fill] at (1.5,0) {};
\node [draw,shape=circle,inner sep=1.5pt,fill] at (1.5,-1) {};
\node [draw,shape=circle,inner sep=1.5pt,fill] at (3,-1) {};
\node [draw,shape=circle,inner sep=1.5pt,fill] at (3,0) {};
\node [draw,shape=circle,inner sep=1.5pt,fill] at (3,1) {};
\node [draw,shape=circle,inner sep=1.5pt,fill] at (3,2) {};
\node [draw,shape=circle,inner sep=1.5pt,fill] at (3,3) {};
\node [draw,shape=circle,inner sep=1.5pt,fill] (v13) at (4.5,3) {};
\node [draw,shape=circle,inner sep=1.5pt,fill] (v16) at (4.5,2) {};
\node [draw,shape=circle,inner sep=1.5pt,fill] (v18) at (4.5,1) {};
\node [draw,shape=circle,inner sep=1.5pt,fill] (v20) at (4.5,0) {};
\node [draw,shape=circle,inner sep=1.5pt,fill] (v22) at (4.5,-1) {};
\node (v2) at (-3,3.6) {};
\node (v1) at (-3,-1.6) {};
\node (v3) at (-1.5,-1.6) {};
\node (v4) at (-1.5,3.6) {};
\node (v6) at (0,3.6) {};
\node (v5) at (0,-1.6) {};
\node (v7) at (1.5,-1.6) {};
\node (v8) at (1.5,3.6) {};
\node (v10) at (3,3.6) {};
\node (v9) at (3,-1.6) {};
\node (v11) at (4.5,-1.6) {};
\node (v12) at (4.5,3.6) {};
\draw [blue]  (v1) edge (v2);
\draw  [blue](v3) edge (v4);
\draw  [blue](v5) edge (v6);
\draw  [blue](v7) edge (v8);
\draw [blue] (v9) edge (v10);
\draw [blue] (v11) edge (v12);
\draw [blue, dotted,thick] (v13) edge (v14);
\draw [blue, dotted, thick] (v15) edge (v16);
\draw [blue, dotted, thick] (v17) edge (v18);
\draw [blue, dotted, thick] (v19) edge (v20);
\draw [blue, dotted, thick] (v21) edge (v22);
\end{tikzpicture}
\end{center}
\caption{\small The local fields and the couplings along the spatial direction can be either positive or negative but have magnitudes reduced by $1/L$ (dotted lines), whereas the couplings along the imaginary-time direction can have magnitude growing logarithmically with system size but are always ferromagnetic (solid lines).}\label{fig:anisotropic}
% TODO: Annotate the figure along the lines of the caption
\end{figure}
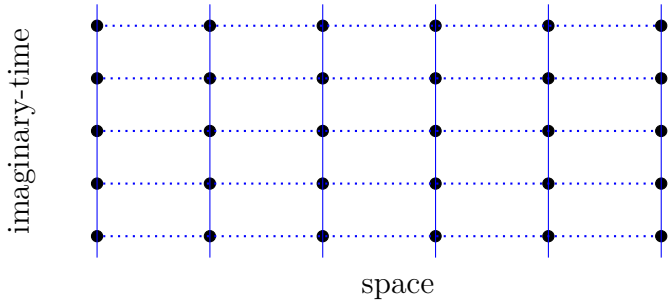
The PIMC method proceeds by attempting to sample from $\pi$ and approximate $Z$ using the Markov chain Monte Carlo method with single-site Metropolis transition probabilities.  With this in mind, the primary question is to determine the \emph{mixing time} $\tau$ for this Markov chain, which is the minimum number of steps of this Markov chain needed to sample from a distribution that is close to $\pi$.  We will use the method of canonical paths that was originally introduced by Jerrum and Sinclair (reviewed in \cite{jerrum-1996,peres-2008}) in order to bound the mixing time.  The idea of the method is that for every pair of states $\mathbf{x},\mathbf{y}\in \Omega$ the Markov chain $P$ needs to route an amount of ``traffic'' equal to $\pi(\mathbf{x})\pi(\mathbf{y})$ along some path $\gamma_{\mathbf{x},\mathbf{y}}$ which connects $\mathbf{x}$ to $\mathbf{y}$ using the transitions of the Markov chain.   Define the {\em congestion} across a particular transition $e = (\mathbf{z},\mathbf{z}')$ to be the ratio of the total traffic through it (i.e. $\sum_{\mathbf{x},\mathbf{y} : e \in \gamma_{\mathbf{x},\mathbf{y}}} \pi(\mathbf{x})\pi(\mathbf{y})$) to the amount of probability flow across it at equilibrium, which equals $\pi(\mathbf{z})P(\mathbf{z},\mathbf{z}')$.  The mixing time can then be upper bounded in terms of the maximum congestion through any transition, using standard arguments which we review in \secref{markovchains}.

Our problem now reduces to finding good sets of paths between every pair of configurations $\mathbf{x},\mathbf{y}\in\Omega$.  We use a set of paths inspired by those used to analyze the (unweighted) random walk on the boolean hypercube~\cite{jerrum-1996}.  Specifically we will replace the entries of $\mathbf{x}$ with entries of $\mathbf{y}$ one at a time.  A crucial requirement is that the energies of the intermediate states do not get too large, as this would create a large congestion.  For the model \eqref{eq:timenergy} the paths first update the classical spins corresponding to the site of qubit 1 in consecutive imaginary-time order, and then proceed to do the same for the sites of qubits 2 through $n$.  For each qubit $j \in \{1,\ldots,n\}$ the line of sites $\{(1,j),(2,j),\ldots,(L,j)\}$ is called the worldline of the $j$-th qubit, and so these canonical paths have the form of updating each qubit worldline consecutively, always finishing one worldline before starting the next.
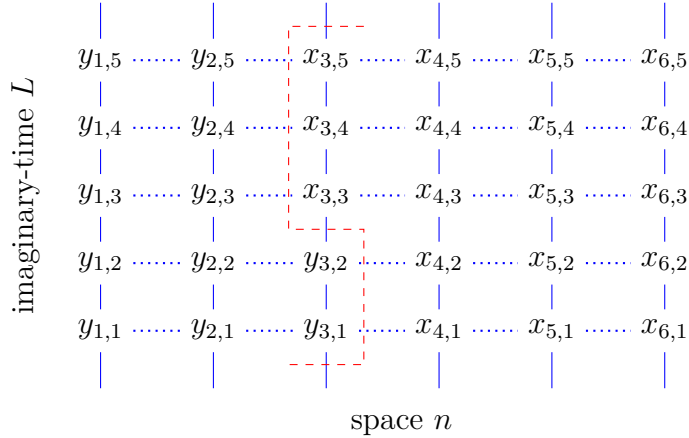
\begin{figure}
\captionsetup{width=0.89\textwidth}
\begin{center}
\begin{tikzpicture}[xscale = 1,yscale=0.9]
\node (v14) at (-3,3) {\large $y_{1,5}$};
\node (v15) at (-3,2) {\large $y_{1,4}$};
\node (v17) at (-3,1) {\large $y_{1,3}$};
\node (v19) at (-3,0) {\large $y_{1,2}$};
\node (v21) at (-3,-1) {\large $y_{1,1}$};
\node (v23) at (-1.5,-1) {\large $y_{2,1}$};
\node (v24) at (-1.5,0) {\large $y_{2,2}$};
\node (v25) at (-1.5,1) {\large $y_{2,3}$};
\node (v26) at (-1.5,2) {\large $y_{2,4}$};
\node (v27) at (-1.5,3) {\large $y_{2,5}$};
\node (v28) at (0,-1) {\large $y_{3,1}$};
\node (v29) at (0,0) {\large $y_{3,2}$};
\node (v30) at (0,1) {\large $x_{3,3}$};
\node (v31) at (0,2) {\large $x_{3,4}$};
\node (v32) at (0,3) {\large $x_{3,5}$};
\node (v37) at (1.5,3) {\large $x_{4,5}$};
\node (v36) at (1.5,2) {\large $x_{4,4}$};
\node (v35) at (1.5,1) {\large $x_{4,3}$};
\node (v34) at (1.5,0) {\large $x_{4,2}$};
\node (v33) at (1.5,-1) {\large $x_{4,1}$};
\node (v38) at (3,-1) {\large $x_{5,1}$};
\node (v39) at (3,0) {\large $x_{5,2}$};
\node  (v40) at (3,1) {\large $x_{5,3}$};
\node (v41) at (3,2) {\large $x_{5,4}$};
\node  (v42) at (3,3) {\large $x_{5,5}$};
\node (v13) at (4.5,3) {\large $x_{6,5}$};
\node  (v16) at (4.5,2) {\large $x_{6,4}$};
\node (v18) at (4.5,1) {\large $x_{6,3}$};
\node (v20) at (4.5,0) {\large $x_{6,2}$};
\node (v22) at (4.5,-1) {\large $x_{6,1}$};
\node (v2) at (-3,4) {};
\node (v1) at (-3,-2) {};
\node (v3) at (-1.5,-2) {};
\node (v4) at (-1.5,4) {};
\node (v6) at (0,4) {};
\node (v5) at (0,-2) {};
\node (v7) at (1.5,-2) {};
\node (v8) at (1.5,4) {};
\node (v10) at (3,4) {};
\node (v9) at (3,-2) {};
\node (v11) at (4.5,-2) {};
\node (v12) at (4.5,4) {};
\draw [blue] (v1) -- (v21) -- (v19) -- (v17) -- (v15) -- (v14) -- (v2);
\node at (1,-2.4) {\large space $n$};
\node [rotate=90] at (-4,1) {\large imaginary-time $L$};
\draw [blue](v3) -- (v23) -- (v24) -- (v25) -- (v26) -- (v27) -- (v4);
\draw [blue](v5) -- (v28) -- (v29) -- (v30) -- (v31) -- (v32) -- (v6);
\draw [blue](v7) -- (v33) -- (v34) -- (v35) -- (v36) -- (v37) -- (v8);
\draw [blue](v9) -- (v38) -- (v39) -- (v40) -- (v41) -- (v42) -- (v10);
\draw [blue](v11) -- (v22) -- (v20) -- (v18) -- (v16) -- (v13) -- (v12);
\node (v43) at (0.5,3.5) {};
\node (v44) at (-0.5,3.5) {};
\node (v45) at (-0.5,0.5) {};
\node (v46) at (0.5,0.5) {};
\node (v47) at (0.5,-1.5) {};
\node (v48) at (-0.5,-1.5) {};
\draw [inner sep = 0pt,red, dashed](v43.center) -- (v44.center) -- (v45.center) -- (v46.center) -- (v47.center) -- (v48.center);
\draw [blue, dotted, thick](v14) -- (v27) -- (v32) -- (v37) -- (v42) -- (v13);
\draw [blue, dotted, thick](v15) -- (v26) -- (v31) -- (v36) -- (v41) -- (v16);
\draw [blue, dotted, thick](v17) -- (v25) -- (v30) -- (v35) -- (v40) -- (v18);
\draw [blue, dotted, thick](v19) -- (v24) -- (v29) -- (v34) -- (v39) -- (v20);
\draw [blue, dotted, thick](v21) -- (v23) -- (v28) -- (v33) -- (v38) -- (v22);
\end{tikzpicture}
\end{center}
\caption{\small An intermediate point along a canonical path from $x$ to $y$.  We have finished updating qubits 1 and 2 and are in the middle of updating qubit 3.  The red line indicates the domain wall formed between the $x$ and $y$ configurations.  The change in the classical energy created by this domain wall determines the Markov chain congestion for this set of canonical paths.}
\label{fig:path}
\end{figure}

An example of an intermediate step along this path is depicted in \fig{path}. Notice that intermediate steps have only introduced a constant number of broken bonds between spins that are neighbors in the imaginary time direction, each contributing an amount of energy that is $\cO(\beta^{-1}\log n)$, and at most $L$ broken bonds between spins that are neighbors in the space direction, each contributing energy $\cO(1/L)$. The relatively low energy of the intermediate configurations along the path is the key to the formal proof of rapid mixing for PIMC applied to 1D generalized TIM in \secref{transverseIsing}.  For 1D generalized TIM with long-range interactions, the energy of the intermediate configuration $\mathbf{z}$ will include contributions from spins that are arbitrarily far away from the cut.  Using the assumption on the rate of decay of the magnitude of these interactions, the total contribution to the energy is upper bounded by
$$
\sum_{j = 1}^{n/2} \sum_{k = \frac{n}{2} + 1}^{n} |j - k|^{-(2 + \xi)}  = \mathcal{O}(\xi^{-2}  + n^{-\xi}),
$$
where the bound is proven by viewing the summation as a right Riemann sum of the corresponding integral, which yields an upper bound because the integrand is a monotonically decreasing function.  Since we regard $\xi$ as a constant this contribution to the energy is $\mathcal{O}(1)$.  

For Hamiltonians of the form \eqref{eq:xxham} the 2-local off-diagonal terms give rise to a classical energy function with interactions involving four spins at a time,
\begin{align}
E_\beta(\mathbf{z}) := \sum_{(i,j) \in \Lambda} \left( \frac{K^{zz}_j}{L} z_{i,j} z_{i,j+1} + \frac{K^z_j}{L} z_{i,j} - \beta^{-1} h_{i,j}\left(z_{i,j},z_{i,j+1},z_{i+1,j},z_{i+1,j+1}\right) \right) \label{eq:generalenergy} ,
\end{align}
%where $\mathbf{z}(\diamond_{i,j}) := \left(z_{i,j},z_{i,j+1},z_{i+1,j},z_{i+1,j+1}\right)$ for $\mathbf{z} \in \Omega$.  These 4-body terms $h_j$ can be expressed by cumbersome analytical formulas but for our purpose it suffices to use the somewhat implicit definition \eqref{eq:aitchjay} (in Appendix~\ref{sec:qtc}) and to upper and lower bound the magnitude of these terms.  
The $h_{i,j}$ terms (which are implicitly defined by the Gibbs distribution \eqref{eq:generalQtCpartitionFunction}) correspond to strong ferromagnetic interactions between neighboring pairs of spins in neighboring time slices, and violating them increases the energy by an amount that is $\mathcal{O}(\log n)$.  If $z_{i,j} \neq z_{i+1,j}$ and $z_{i,j+1} \neq z_{i+1,j+1}$ then we say there is a 2-local jump between time slice $i$ and time slice $i+1$, because the resulting energy contribution from the 4-body interaction will depend on the magnitude of the 2-local off-diagonal quantum terms i.e. $K^{xx}_j$ and $K^{yy}_j$.

The canonical paths we have described so far will lead to intermediate configurations that potentially create or destroy a non-constant number of 2-local jumps across the domain wall pictured in figure~\ref{fig:path}.  If $K^{xx}_j$ is larger than $K_j^x$, then configurations $\mathbf{x}$ and $\mathbf{y}$ which contain a non-constant number of 2-local jumps will result in the canonical paths described so far overloading some edges with large congestion.  To overcome this we include additional path segments along the canonical paths that remove 2-local jumps from the worldlines near the domain wall in figure~\ref{fig:path}.  This greatly increases the number of paths passing through each edge, but we are still able to obtain a $\poly(n,e^\beta)$ upper bound on the congestion by noticing that most of these paths have endpoints with very small stationary probabilities.

For Hamiltonians of the form \eqref{eq:Ham} it is possible that the off-diagonal terms in \eqref{eq:Ham} do not allow for the PIMC Markov chain with single-site transitions to be ergodic.  In this situation some configurations $\mathbf{z} \in \Omega$ will have $\pi(\mathbf{z}) = 0$, which creates a new obstacle to the canonical paths we have described so far.  To restore ergodicity under single-site updates we add a fictitious transverse field to the Hamiltonian with a weak coupling that goes to zero with system size:
\be
H \rightarrow H - \Gamma \sum_{i=1}^n \sigma^x_i
\ee
with $\Gamma \leq \delta_{\text{mult}}/n$ so that the subsequent estimate for the partition function will not be further than the allowed error tolerance away from its intended value. 

Once the fictitious field is added it once again makes sense to attempt to use the same canonical paths that were applied to systems of the form \eqref{eq:xxham}.  But now the fact that the off-diagonal matrix elements of the Hamiltonian are no longer symmetric under the global operation of flipping all spins means that removing a single (1-local or 2-local) jump from a given worldline can decrease the stationary probability of that configuration by a factor scaling exponentially in the number of jumps remaining in that worldline.  To overcome this problem we show in \secref{restricting} that typical configurations distributed according to $\pi$ will have at most $\mathcal{O}(\beta \log n)$ jumps in any particular worldline.  This means that the worst-case intermediate configurations just described may have a stationary probability reduced by $n^{O(\beta \log n)}$, and this factor bottlenecks the mixing time estimate.  Using these facts we describe canonical paths in \secref{general1D} within a restricted state space $\Omega^*$ in which every configuration has $\cO(\beta \log n)$ jumps per worldline, and by either restricting the PIMC Markov chain to $\Omega^*$ or applying the leaky random walk conductance bound~\cite{crosson2016simulated} we obtain a bound on the mixing time within the subset $\Omega^*$ that establishes theorem~\ref{theo:tim}.

\section{Path integral Monte Carlo} \label{sec:PIMCcorrect}
Sections \ref{sec:qtc} through \ref{sec:estimatePartitionFunction} rigorously present the underpinnings of the standard PIMC method: the quantum-to-classical mapping, and its application to approximating expectation values of observables as well as the partition function.  \secref{restricting} proves a general concentration of measure bound which shows that the number of jumps along any worldline in PIMC is $\mathcal{O}(\beta \log n)$ with high probability.    Note that this result applies to general stoquastic Hamiltonians (not just the 1D case).  This result is used to analyze a restricted state space in \secref{general1D}.
\subsection{Quantum-to-classical mapping}\label{sec:qtc}
The first step of the quantum-to-classical mapping is to split $H$ as in \eqref{eq:Ham} into a sum of commuting Hamiltonians.   For each local term $H_i$ we construct local operators $A_i, B_i, C_i$ such that $-\beta H_i = A_i + B_i + C_i$ and $[A_i,A_j] = [B_i, B_j] = [C_i,C_j]=0$ for all $i,j$.  Let $A = \sum_i A_i$, $B= \sum_i B_i$, $C = \sum_i C_i$, and hereafter let $A_i$ match the diagonal part of $H_i$ in the computational basis, while $B_i$ and $C_i$ contain the off-diagonal terms of $H_i$ when $i$ is odd or even, respectively.    For positive integers $L$ define the Suzuki-Trotter approximation of the partition function $Z = \tr\left[e^{A + B + C}\right]$,
\begin{equation}
Z_{\beta,L} := \tr\left [\left(e^{\frac{A}{2L}} e^{\frac{B}{L}}e^{\frac{A}{2L}}e^{\frac{C}{L}}\right)^L\right].\label{eq:part-Trotter}
\end{equation}
The fact that $\lim_{L\ra \infty}Z_{\beta,L} = \cZ_\beta$ is the content of the 1875 Lie product formula, while for finite $L$ the Suzuki-Trotter approximation is related to the quantum partition function \eq{Part} by
\be 
e^{-\delta} \cZ_{\beta} \leq
Z_{\beta, L}  \leq e^{\delta}\cZ_{\beta} ,
\label{eq:part-approx}
\ee
for some $\delta$ that is $\cO\left(\|H\|^3/L^3\right)$.  This error bound follows from two facts about the matrix exponential: a variant of the usual Suzuki-Trotter expansion, and a continuity bound.  Both appear in \cite{bravyi-2014} for transverse Ising models, and the proof for three commuting layers $A,B,C$ is similar and so we omit the details.
%and many results in \cite{suzuki-1977} are derived in this limit when the Suzuki-Trotter form becomes exact.  However, the algorithm described in this work operates at a fixed $L$ for each input Hamiltonian,  which introduces an error%\footnote{Computational techniques for extrapolating results from PIMC to the $L\rightarrow\infty$ limit do exist, so the use of the Suzuki-Trotter approximation should not be considered an essential drawback of the method. }

Before considering the general case of $H$ as in \eq{Ham} it is useful to consider the simplified case of \eq{part-Trotter} for generalized TIM of the form \eq{tim} by taking $A = -\beta \diag(H)$, $B = -\beta\left(H - \diag(H)\right)$, and $C = 0$ so that the Suzuki-Trotter approximation reduces to
\begin{equation}
Z_{\beta,L} = \tr\left [\left(e^{\frac{A}{2L}}
e^{\frac{B}{L}}e^{\frac{A}{2L}}
\right)^L\right]= \tr\left [\left(e^{\frac{A}{L}}
e^{\frac{B}{L}}
\right)^L\right].\label{eq:part-Trotter-tim}
\end{equation}

For transverse Ising models, the next step is to expand the Suzuki-Trotter approximation \eqref{eq:part-Trotter-tim} to the partition function by inserting several complete sets of states,
\begin{align}
Z_{\beta, L} &= \sum_{z_1\in \{-1,1\}^n} \langle z_1 | \left(e^{\frac{A}{L}} e^{\frac{B}{L}}\right)^L|z_1\rangle \label{QtC1}\\
&= \sum_{z_1,\ldots,z_L} \prod_{i=1}^{L}\langle z_{i}| e^{\frac{A}{L}} e^{\frac{B}{L}}|z_{i+1}\rangle \; \; \; , \; \; z_{L+1} := z_1\\
&= \sum_{z_1,\ldots,z_L} \prod_{i=1}^L e^{\frac{A(z_{i})}{L}}\langle z_{i}| e^{\frac{B}{L}}|z_{i+1} \rangle \label{QtC3}.
\end{align}
where in the last step we have used the fact that $A = -\beta \diag(H)$ is diagonal in the computational basis and defined $A(z_i) := \langle z_i | A | z_i\rangle$.  The Gibbs distribution of the classical 2D spin system is defined by equating $e^{- \beta E_{\beta}(z_1,\ldots,z_L)}$ with the summand of \eqref{QtC3}.  The closed form \eqref{eq:timenergy} is obtained applying the identity $e^{\omega \sigma^x} = \cosh(\omega)I + \sinh(\omega)\sigma^x$ to
$$
\langle z_{i}| e^{\frac{B}{L}}|z_{i+1}\rangle  = \langle z_{i}| e^{-\frac{\beta \Gamma}{L} \sum_{j=1}^n \sigma^x_j}|z_{i+1}\rangle = \prod_{j=1}^n \langle z_{i,j}|e^{-\frac{\beta \Gamma}{L} \sigma^x_j}|z_{i+1,j}\rangle. \label{eq:tim-qtc-offdiag}
$$
For a general 1D stoquastic Hamiltonian of the form in \eq{Ham} we will apply \eq{part-Trotter} with
\be
A = -\beta \diag(H)  \; \; \; , \; \; \; B = -\beta \sum_{i \in \textnormal{Odd}} \L[H_i - \diag(H_i) \R]  \; \; \; , \; \; \; C = -\beta \sum_{i\in \textnormal{Even}}\L[ H_i - \diag(H_i) \R],
\ee
Expanding \eq{part-Trotter} with $2 L$ complete sets of states,
\begin{align}
Z_{\beta,L} &= \sum_{z_1\in \{-1,1\}^n} \langle z_1| \left(e^{\frac{A}{2L}} e^{\frac{B}{L}}
e^{\frac{A}{2L}}e^{\frac{C}{L}}\right)^L | z_1 \rangle\\
&= \sum_{z_1,\ldots,z_{2L}} \prod_{i\in \textrm{Odd}}^{2L - 1}\langle z_i |e^{\frac{A}{2L}} e^{\frac{B}{L}}|z_{i+1}\rangle \langle z_{i+1}|e^{\frac{A}{2L}}e^{\frac{C}{L}} | z_{i+2} \rangle \label{eq:qtcGen2}\\
&=\sum_{z_1,\ldots,z_{2L}} e^{\sum_{i=1}^{2L} \frac{A(z_i)}{2 L} }\prod_{i\in \textrm{Odd}}^{2L - 1}\langle z_i |e^{\frac{B}{L}}|z_{i+1}\rangle \langle z_{i+1}|e^{\frac{C}{L}}| z_{i+2} \rangle\label{eq:qtcGen3},
\end{align}
where periodic boundary conditions $z_{2 L +1} := z_1$ are taken in \eq{qtcGen2}.  At this point the essential difference from the transverse Ising case is clear: the off-diagonal part of the Hamiltonian is split into two commuting layers, these layers alternate along the imaginary-time direction, and this causes the coupling between classical spins at neighboring time slices to alternate as well.   For ease of notation define $B_i$ to be $2B$ if $i$ is odd, and $2C$ if $i$ is even, and rescale $L\rightarrow L/2$ to obtain a distribution $\pi$ from \eq{qtcGen3},
\be
\pi(\mathbf{z}) := \frac{1}{Z_{\beta,L}}e^{\sum_{i=1}^{L} \frac{A(z_i)}{L} }\prod_{i = 1}^{L}\langle z_i |e^{\frac{B_i}{L}}|z_{i+1} \rangle\label{eq:generalQtCpartitionFunction} .
\ee
\subsection{Approximating observables}\label{sec:approxObservables} 
%%%%%%%%%%%%%%%%

%We would also like to estimate expectation values of observables.
{In this section we consider an observable $\cO$ for which we would like to estimate $\vev{\cO}_\beta := \tr[\cO\rho_\beta]$. We justify this by demonstrating closeness of $\vev{\cO}_\beta$ to a value $\langle O\rangle_{\beta,L}$ that is defined in terms of the Suzuki-Trotter approximation in equation \eqref{eq:st-observable}.  } The expectation $\vev{\cO}_\beta$ can be expressed as the derivative of a partition function by defining
\be \cZ_\beta(\zeta) := \tr[e^{-\beta \cH + \zeta \cO}],\label{eq:Z-zeta}\ee
and using the fact that 
\be
\vev{\cO}_\beta = \frac{1}{ \cZ_\beta(0)}\frac{\partial  \cZ_\beta}{\partial \zeta}(0).
\ee
The next proposition shows that $\langle \cO \rangle_{\beta}$ can be approximated by a finite difference.
\begin{prop}
Suppose that $\hat{\cZ}_\beta(0)$ and $\hat{\cZ}_\beta(\zeta)$ satisfy
$$e^{-\delta} \leq \frac{\hat{\cZ}_\beta(0)}{{\cZ}_\beta(0)},
\frac{\hat{\cZ}_\beta(\zeta)}{{\cZ}_\beta(\zeta)} \leq e^{\delta},$$
for $\zeta = \sqrt{3\delta}/\|\cO\|$.
Then the estimator
\be \widehat{\vev{\cO}_\beta}  := \frac{\hat\cZ_\beta(\zeta) - \hat\cZ_\beta(0)}{\zeta\hat\cZ_\beta(0)}.\label{eq:finiteDiff}\ee
satisfies the bound
\be
|\widehat{\vev{\cO}_\beta} - {\vev{\cO}_\beta}|
\leq 2\sqrt{\delta} \|\cO\|\ee
\end{prop}

\begin{proof}
We use first Taylor's theorem, then the bound $\tr AB \leq \|A\| \|B\|_1$ to obtain: %and finally \lemref{exp-Lip}
\begsub{discrete-zeta}
\L|\cZ_\beta(\zeta) - \cZ_\beta(0) - \zeta \frac{\partial  \cZ_\beta}{\partial \zeta}(0)\R| 
&\leq \frac{\zeta^2}{2} \max_{\zeta' \in [0,\zeta]} \frac{\p^2\cZ_\beta}{\p\zeta^2}(\zeta')
\\ & \leq  \zeta^2 \|\cO\|^2 
\max_{\zeta' \in [0,\zeta]} \cZ_\beta(\zeta')
\\ & \leq \zeta^2 \|\cO\|^2 \cZ_\beta(0) e^{\zeta \|\cO\|}.
\endsub
Dividing both sides of \eq{discrete-zeta} by $\zeta\cZ_\beta(0)$ we  have
\be \L| \vev{\cO}_\beta - \frac{\cZ_\beta(\zeta)-\cZ_\beta(0)}{\zeta \cZ_\beta(0)} \R|=
\L| \vev{\cO}_\beta -\L( \frac{\cZ_\beta(\zeta)}{\zeta \cZ_\beta(0)} - \frac 1 \zeta\R) \R|
\leq \zeta \|\cO\|^2  e^{\zeta \|\cO\|}.
\label{eq:discrete-O}\ee

Then
\ba
\widehat{\vev{\cO}_\beta} - {\vev{\cO}_\beta}
& = \frac{\hat\cZ_\beta(\zeta)}{\zeta\hat\cZ_\beta(0)} - \frac 1 \zeta - \vev{\cO}_\beta \\
& \leq e^{2\delta} \frac{\cZ_\beta(\zeta)}{\zeta\cZ_\beta(0)} - \frac 1 \zeta - \vev{\cO}_\beta \\ 
& \leq (e^{2\delta}-1) \frac{\cZ_\beta(\zeta)}{\zeta\cZ_\beta(0)}
+ \zeta \|\cO\|^2  e^{\zeta \|\cO\|}\\
& \leq (3\delta/\zeta + \zeta \|\cO\|^2)  e^{\zeta \|\cO\|}\\
& \leq 2\sqrt{\delta} \|\cO\|
\ea
In the last step we have
used the assumption that $\delta \leq 1/21$ to simplify the numerical constants.
A similar argument shows that
$\widehat{\vev{\cO}_\beta} - {\vev{\cO}_\beta} \geq -2\sqrt\delta \|\cO\|$, thus completing the proof.
\end{proof}
\noindent  By \eq{finiteDiff}, \eq{part-approx}, and the triangle inequality, the expectation $\langle \cO\rangle_\beta$ can be  approximated by
\begin{equation}
\langle O \rangle_{\beta, L} = \frac{1}{Z_{\beta,L(0)}}\frac{\partial}{\partial \zeta} Z_{\beta,L}(0), \label{eq:st-observable}
\end{equation}
where $Z_{\beta,L}(\zeta)$ is a Suzuki-Trotter approximation to \eq{Z-zeta}.  If $\cO$ is diagonal use the form in \secref{qtc} with $A \rightarrow A + \zeta \cO$, 
\be
 Z_{\beta,L}(\zeta)  = \sum_{\mathbf{z}\in \Omega} e^{\frac{1}{L} \sum_{i=1}^{L} \L(A(z_i) + \zeta \cO(z_i)\R)}\prod_{i = 1}^{L}\langle z_i |e^{\frac{\cB_i}{L}}|z_{i+1}\rangle.
\ee
Evaluating the derivative at $\zeta = 0$,  
\begin{align}
\frac{1}{Z_{\beta,L}(0)} \frac{\partial}{\partial\zeta} Z_{\beta, L}(0) &= \frac{1}{Z_{\beta,L}} \sum_{\mathbf{z}\in \Omega}\L (\frac{1}{L}\sum_{i=1}^{L}\cO(z_i)\R)e^{\frac{1}{L} \sum_{i=1}^{L} A(z_i)} \prod_{i = 1}^{L}\langle z_i |e^{\frac{\cB_i}{L}}|z_{i+1}\rangle\\
&= \sum_{\mathbf{z}\in \Omega}\pi(\mathbf{z})\L(\frac{1}{L}\sum_{i=1}^{L}\cO(z_i) \R) = \frac{1}{L}\sum_{i = 1}^{L} \langle O(z_i) \rangle_{\pi} =\langle \cO(z_1) \rangle_{\pi}
\end{align}
where the final expression follows by the symmetry of $\pi$ under cyclic permutations of the time slices.  For off-diagonal $\cO$, 
\be
 Z_{\beta,L}(\zeta)  = \sum_{\mathbf{z}\in \Omega} e^{\frac{1}{L} \sum_{i=1}^{L} A(z_i) } \prod_{i = 1}^{L}\langle z_i |e^{\frac{1}{L}\L(\cB_i + \zeta O/2\R)}|z_{i+1}\rangle.
\ee
Evaluating the derivative at $\zeta = 0$,  
\ba
\langle O \rangle_{\beta,L} &= \frac{1}{Z_{\beta,L}(0)} \sum_{\mathbf{z} \in \Omega} e^{\frac{1}{L}\sum_{i =1}^{L} A(z_i)} \sum_{k=1}^{L} \langle z_k|\L(\frac{O}{2L}\R) e^{\frac{\cB_k}{L}}|z_{k+1}\rangle \prod_{\substack{i = 1\\i \neq k}}^{L} \langle z_i | e^{\frac{\cB_i}{L}} | z_{i+1} \rangle  \\ 
&= \sum_{\mathbf{z} \in \Omega}\L( \frac{1}{2L} \sum_{k=1}^{L} \frac{\langle z_k|O e^{\frac{\cB_k}{L}}|z_{k+1}\rangle}{\langle z_k|e^{\frac{\cB_k}{L}}|z_{k+1}\rangle } \R) \pi(\mathbf{z}) \\
&= \frac{1}{2L} \sum_{k=1}^{L} \L\langle \frac{\langle z_k|O e^{\frac{\cB_k}{L}}|z_{k+1}\rangle}{\langle z_k|e^{\frac{\cB_k}{L}}|z_{k+1}\rangle }\R\rangle_{\pi}\\
&=\frac{1}{2}\sum_{\mathbf{z}\in \Omega} \L [ \frac{\langle z_1 | \cO e^{\frac{B_1}{L}} | z_2 \rangle}{\langle z_1 |e^{\frac{B_1}{L}} | z_2 \rangle} + \frac{\langle z_2 | \cO e^{\frac{B_2}{L}} | z_3 \rangle}{\langle z_2 |e^{\frac{B_2}{L}} | z_3 \rangle} \R ] \pi(\mathbf{z})\label{eq:observableCompute}
\ea 
where the last step follows by symmetry of $\pi$ under even cyclic permutations of the time slices.  Since the summand of \eq{observableCompute} is efficiently computable and has bounded magnitude we can estimate it using a standard Monte Carlo procedure (see lemma~\ref{lem:hoeff}) once one has an efficient meanns of obtaining independent samples from $\pi$.

\subsection{Approximating partition functions}\label{sec:estimatePartitionFunction}
In Section~\ref{sec:proof} we will prove that for any tolerance $\delta > 0$ the PIMC method can efficiently output an element of
$\Omega$ sampled from a distribution $\tilde{\pi}$ which is within
$|\tilde{\pi} - \pi| \leq \delta$.  Here we explain how these samples can be used to approximate the partition function with an application of the standard telescoping product technique \cite{jerrum-1993, vigoda-2005}\footnote{{The techniques used in this section to estimate the partition function (or equivalently the free energy) are necessary because the entropy term in the free energy cannot be directly estimated as an observable.  If one is only interested in the estimation of specific observables then it suffices to use the end result\eqref{eq:observableCompute}}.}.  We express the partition function $Z(\beta)$ as a product of terms involving the partition function at higher temperatures, 

\begin{equation}
Z(\beta) = Z(0) \frac{Z(\beta_1)}{Z(\beta_0)} \ldots \frac{Z(\beta_k)}{Z(\beta_{k-1})}
\end{equation} 
where the temperature schedule $\beta_0 < \beta_1 <\ldots<\beta_k$ with $\beta_0 = 0$ and $\beta_k = \beta$ is chosen so that each ratio in the telescoping product is bounded.  Defining $w_{\beta_i}(\mathbf{z}) : = Z(\beta_i) \pi_{\beta_i}(\mathbf{z})$, note that
\begin{equation}
\left \langle \frac{w_{\beta_i}}{w_{\beta_{i-1}}} \right \rangle_{\pi_{\beta_{i-1}}} = \sum_{\mathbf{z}\in \Omega} \frac{w_{\beta_i}(\mathbf{z})}{w_{\beta_{i-1}}(\mathbf{z})}\pi_{\beta_{i-1}}(\mathbf{z}) = \frac{Z(\beta_i)}{Z(\beta_{i-1})}.
\end{equation}
At infinite temperature the quantum partition function is $\mathcal{Z}(0) = 2^n$.  Further, a uniform schedule with $k = \mathcal{O}(\beta \|H\| \log(\beta \|H\|))$ ensures that
\begin{equation}
\frac{1}{e} \leq  \frac{\mathcal{Z}(\beta_i)}{\mathcal{Z}(\beta_{i-1})} \leq 1.
\end{equation}
%and so together with the results in \secref{suzuki-Trotter} this implies the Suzuki-Trotter approximation to the partition function satisfies
%\begin{equation}
%e^{-(1+2\delta)} \leq  \frac{Z(\beta_i)}{Z(\beta_{i-1})} \leq e^{2\delta}
%\end{equation}
Note that in infinite temperature limit $\beta = 0$ the classical effective systems \eqref{eq:timenergy} and \eqref{eq:generalenergy} have infinite coupling strength between spins in the imaginary-time direction, and so the Markov chain we analyze in \secref{proof} will not be ergodic.  This is ununsual from the perspective of classical spin systems, but it happens in this case because of the temperature dependence of the couplings.  To deal with this the first expectation $\left \langle w_{\beta_1}/w_{\beta_{0}} \right \rangle_{\pi_{\beta_{0}}}$ in the telescoping product can be computed using the fact that $\pi_{\beta_0}$ is equivalent to the uniform distribution on the subset $\Omega_{\textnormal{classical}} \subseteq \Omega$ consisting of configurations $(z_1,\ldots,z_L)$ which have $z_i = z_j$ for all $i,j = 1,\ldots,L$.  

Finally, in order to bound the number of samples needed to compute the expectation values in the telescoping product we will make use of the Hoeffding bound.
\begin{lemma}[\cite{Hoeffding-1963}]\label{lem:hoeff}
Let $X_1,\ldots,X_t$ be independent random variables satisfying $|X_i|
\leq 1$ and $\bbE[X_i] = \bar X$.  Then 
\be \bbP[\L| \frac{1}{t} \sum_{i=1}^t X_i
  - \bar X \R| \geq \delta] \leq 2 e^{-t\delta^2/2}\ee
\end{lemma}
Define $M(z) = e^{-2 \delta} w_{\beta_i}(z)/w_{\beta_{i+1}}(z)$, so that $|M(z)| \leq 1$. Observe that $|\langle M \rangle_{\pi} - \langle M \rangle_{\tilde\pi}| \leq \|\pi - \tilde\pi\|_1 \|M\|
\leq \delta$. We will estimate $\langle M \rangle_{\tilde\pi}$ by drawing $t$ samples from $\tilde \pi$,
which we call $\bz^{(1)},\ldots,\bz^{(t)}$.  Our estimator will simply
be the sample mean, namely
\be \hat{\bar M} := \frac{1}{t} \sum_{i=1}^t M(\bz_i).\ee
By \lemref{hoeff} we have 
\be\bbP[|\hat{\bar M} -\langle M \rangle_{\tilde\pi}| \geq \delta] \leq
2e^{-t\delta^2/2}. \label{eq:highprob}\ee
We conclude that taking $t = O(1/\delta^2)$ samples yields an $O(\delta)$
approximation to $\langle M \rangle_{\tilde\pi}$ with high probability.   Therefore a polynomial number of samples will suffice to estimate each term of the telescoping product.

\subsection{Restricting the PIMC state space}\label{sec:restricting}

In this section we justify restricting PIMC to a subset of the state space that limits the number of jumps in any particular worldline to be $\cO(\beta \log n)$ (this restricted state space will be used for the mixing time analysis in \secref{general1D}).  For a configuration $\mathbf{z} = (z_1,...,z_L)$ the number of jumps in the worldline of the $j$-th qubit can be expressed in terms of the Pauli operator $X_j$,
$$
d_j(\mathbf{z}) := |\{i : z_{i,j} \neq z_{i+1,j}\}| =  \sum_{i = 1}^L \langle z_i | X_j | z_{i+1}\rangle.
$$
The strategy is to bound the $k$-th moment of the random variable $d_j$ with respect to the QMC stationary distribution $\pi$ and use the fact that
$$
\Pr[d_j \geq a]_{\pi} \leq \frac{\mathbbm{E}_{\pi} d^k_j}{a^k}.
$$
These moments satisfy
\begin{align}
\mathbbm{E}_{\pi} d^k_j&:= \sum_{\mathbf{z} \in \Omega} d_j^k(\mathbf{z}) \pi(\mathbf{z})\nonumber \\
&= \sum_{\mathbf{z} \in \Omega} \left(\sum_{i = 1}^L \langle z_i | X_j | z_{i+1}\rangle \right)^k \pi(\mathbf{z})\nonumber  \\
&= \sum_{\mathbf{z} \in \Omega} \sum_{i_1,...,i_k = 1}^L  \prod_{r = 1}^k \langle z_{i_r} | X_j | z_{i_r+1}\rangle  \pi(\mathbf{z})\nonumber  \\
&= \frac{1}{Z_{\beta,L}}\sum_{\mathbf{z} \in \Omega} \sum_{i_1,...,i_k = 1}^L  \prod_{r = 1}^k \langle z_{i_r} | X_j | z_{i_r+1}\rangle  \left(e^{\sum_{i=1}^{L} \frac{A(z_i)}{L} }\prod_{i = 1}^{L}\langle z_i |e^{\frac{B_i}{L}}|z_{i+1} \rangle \rangle\right).\label{eq:momentSum}
\end{align}
where $A,  B_i$ are defined in \secref{qtc}.  To merge the two products in the expression above into a simpler trace expression, we want to show the following inequality for any neighboring time slices $z_i$ and $z_{i+1}$, 
\begin{equation}\label{eq:elementwiseInequality}
\langle z_i| X_j |z_{i+1}\rangle \leq 4\frac{\beta}{L} \frac{\langle z_i | X_j e^{B_i/ L} | z_{i+1}\rangle}{\langle z_i | e^{ B_i / L} | z_{i+1}\rangle} ,
\end{equation}
where we assume the local terms are normalized so that for some $\Gamma > 0$
$$\Gamma \leq |\langle z_i | H | z_{i+1}\rangle| \leq 1 \quad \forall z_i,z_{i+1} :
|z_i - z_{i+1}| = 1.$$  If $\langle z_i | X_j | z_{i+1}\rangle = 0$ then \eqref{eq:elementwiseInequality} is satisfied.   Otherwise if $\langle z_i | X_j | z_{i+1}\rangle = 1$ then the numerator on the RHS of \eqref{eq:elementwiseInequality} satisfies
\be
\langle z_{i} | X_j e^{ B_i/ L} | z_{i+1}\rangle 
= \langle z_{i+1} | e^{B_i/ L} | z_{i+1}\rangle 
\geq  e^{-\beta\|H\|/L} \geq \frac 12.
\ee
To upper bound the denominator of the RHS of \eqref{eq:elementwiseInequality} , the fact that $B_i$ is a commuting Hamiltonian implies
\begin{align*}
\langle z_i | e^{B_i/L} |z_{i+1}\rangle &=  \langle z_{i,j} z_{i,j+1}| e^{B_{i,j} /L} |z_{i+1,j},z_{i+1,j+1}\rangle \\
&=\sum_{q = 1}^\infty
\frac{1}{q!} \left\langle z_{i,j} z_{i,j+1}|\left(B_{i,j}/L \right)^q |z_{i+1,j},z_{i+1,j+1} \right\rangle\\
&\leq e^{\frac{\beta \|H_j\|}{L}} - 1\\
&\leq \frac{1}{\ln(2)}\frac{\beta }{L},
\end{align*}
where in the last step we have used $\|H_j\| \leq 1$, ${\frac{\beta}{L}}\leq \ln(2)$
and the convexity of $e^x$.  Using $\ln(2)/2\geq 1/4$ we have established \eqref{eq:elementwiseInequality}.  

To apply this result to \eqref{eq:momentSum} we will use the
permutation symmetry of the $i_1,...,i_k$ variables to restrict to a
sum over ordered $k$-tuples $(i_1,...,i_k)$ which satisfy $i_1 \leq i_2\leq \ldots i_k$.
\begin{align}
\mathbbm{E}_{\pi} d^k_j& \leq \frac{k!}{Z}\sum_{\mathbf{z} \in \Omega} \sum_{1\leq  i_1 \leq ... \leq i_k \leq L}  \prod_{r = 1}^k \langle z_{i_r} | X_j | z_{i_r+1}\rangle \left(e^{\sum_{i=1}^{L} \frac{A(z_i)}{L} }\prod_{i = 1}^{L}\langle z_i |e^{\frac{B_i}{L}}|z_{i+1} \rangle \rangle\right)\\
&\leq  \frac{4^k k! \beta^k}{Z L^k}\sum_{\mathbf{z} \in \Omega} \sum_{1\leq  i_1 \leq ... \leq i_k \leq L}  \prod_{r = 1}^k \frac{\langle z_{i_r} | X_j e^{\frac{B_i}{L}} | z_{i_r+1}\rangle}{\langle z_{i_r} |e^{\frac{B_i}{L}} | z_{i_r+1}\rangle}  \left(e^{\sum_{i=1}^{L} \frac{A(z_i)}{L} }\prod_{i = 1}^{L}\langle z_i |e^{\frac{B_i}{L}}|z_{i+1} \rangle \rangle\right)
\end{align}
Defining $t_r = (i_r - i_{r-1})\beta /L$ (as well as $t_1 = r_1\beta/ L, t_k = (L - i_k)\beta/L$) this expression can be seen as the trace of a product of operators,
\begin{equation}
\mathbbm{E}_{\pi} d^k_j \leq \frac{k! (4\beta)^k}{Z L^k} \sum_{\substack{t_1,...,t_k\\ t_1 + ... + t_k = \beta}}^L  \textrm{tr} \left[M_{t_1}R_{t_1} X_j L_{t_2} M_{t_2} R_{t_2} X_j ... X_j  L_{t_k} M_{t_k} \right]\label{eq:traces}
\end{equation}
where 
$$
M_{t_i} :=  \left( e^{B_{\textrm{even}}/L} e^{A/L} e^{B_{\textrm{odd}}/L} e^{A/L} e^{B_{\textrm{even}}/L}\right)^{\lfloor t_i \beta /L \rfloor -2}
$$  
and $L_i$,$R_{t_i}$ can each denote one of several products of operators depending on $t_i$, which is determined by where the $X_j$ operators disrupt the sequence, 
$$
L_{t_i} =\begin{cases} 
      e^{A/L}  \\
      e^{A/L} e^{B_{\textrm{odd}}/L} \\
      e^{A/L} e^{B_{\textrm{odd }}/L}e^{A/L} \\
     e^{A/L} e^{B_{\textrm{odd}}/L} e^{A/L} e^{B_{\textrm{even}}/L}   
   \end{cases}  \quad , \quad 
   R_{t_i} =    \begin{cases} 
      e^{B_{\textrm{even}}/L}  \\
      e^{A/L} e^{B_{\textrm{even}}/L} \\
      e^{B_{\textrm{odd }}/L}e^{A/L} e^{B_{\textrm{even}}/L} \\
     e^{A/L} e^{B_{\textrm{odd}}/L} e^{A/L} e^{B_{\textrm{even}}/L}   .
   \end{cases}
$$
The purpose of this decomposition is that each $M_{t_i}$ is a PSD operator raised to an positive integer power, and the rest of the Matrices $L_{t_i},R_{t_i},X_j$ are all either PSD operators or are explicit products of 2,3, or 4 PSD operators.  
Now the traces in \eqref{eq:traces} can be bounded using classic inequality due to Ky Fan~\cite[Theorem 1]{Fan51}.  Let $P_1 ... P_N$ satisfy $P_i \succeq 0$ for each $i$, and let $\sigma(P_i)$ be a diagonal matrix with the eigenvalues of $P_i$ along the diagonal in descending order.  Then
$$
\textrm{tr} \left[ P_1 ... P_N \right] \leq \textrm{tr} \left[\sigma(P_1)...\sigma(P_N)  \right]
$$
For our purposes define the notation $\sigma(R_{t_i})$ to be the product of diagonal matrices of eigenvalues of the operators that make up $R_{t_i}$ e.g. if $R_{t_i} = e^{A/L} e^{B_{\textrm{even}}/L}$ then $\sigma(R_{t_i}) = \sigma \left(e^{A/L}\right) \sigma\left( e^{B_{\textrm{even}}/L} \right)$.  Furthermore we may take $B_{\textrm{even}}, B_{\textrm{odd}} \succeq 0$ without loss of generality (by shifting the overall Hamiltonian by a multiple of the identity), so that $\mathbbm{1} \succeq \sigma(R_{t_i}), \sigma(L_{t_i})$ for each $t_i$.  Therefore each trace in \eqref{eq:traces} satisfies
\begin{align}
&\textrm{tr} \left[M_{t_1}R_{t_1} X_j L_{t_2} M_{t_2} R_{t_2} X_j ... X_j  L_{t_k} M_{t_k} \right] \nonumber \\ 
&\leq \textrm{tr} \left[\sigma \left(M_{t_1} \right ) \sigma \left(R_{t_1} \right)\sigma\left(X_j \right) \sigma\left( L_{t_2} \right) \sigma\left( M_{t_2} \right) \sigma \left(R_{t_2}\right) \sigma \left( X_j \right) ...\sigma\left( X_j \right) \sigma\left( L_{t_k} \right) \sigma \left( M_{t_k} \right) \right]\label{eq:sigmaRL1}\\
&\leq \textrm{tr} \left[\sigma \left(M_{t_1} \right ) \sigma\left( M_{t_2} \right)  ... \sigma \left( M_{t_k} \right) \right]
\label{eq:sigmaRL2}\\
&= \textrm{tr} \left[\left( e^{B_{\textrm{even}}/L} e^{A/L} e^{B_{\textrm{odd}}/L} e^{A/L} e^{B_{\textrm{even}}/L}\right)^{L -2k} \right]  \nonumber \\
& = Z_{\beta, L - 2k} \nonumber
\end{align}
where the fact that $\mathbbm{1} \succeq \sigma(R_{t_i}), \sigma(L_{t_i})$ is used to go from \eqref{eq:sigmaRL1} to \eqref{eq:sigmaRL2}.  Now for $k \ll L$ we have $Z_{\beta,L} \approx Z_{\beta, L - 2k}$, so this  implies that $\mathbbm{E}_{\pi} d^k_j \leq k! \beta^k$.   Therefore taking $k=c\log n$ for some $c > 0$ and using $k!\leq (k/e)^k$ yields
$$\Pr_{\pi}[d_j \geq c\beta \log n] \leq \frac{k!\beta^k}{c^k\beta^k\log^k b} = \frac{k!}{k^k} \leq e^{-k} = n^{-c}.$$

\section{Proofs of rapid mixing}\label{sec:proof}
{Having established the correctness of the PIMC method in section \ref{sec:PIMCcorrect} we now procede to analyze the mixing time of the associated Markov chains.  The mixing time analysis is divided into cases.  In section \ref{sec:markovchains} we review the main techniques that are common to each case in our analysis.  Section \ref{sec:transverseIsing} contains the mixing analysis for transverse Ising models, while sections \ref{sec:XX} and \ref{sec:general1D} analyze Hamiltonians of the form \eqref{eq:xxham} and \eqref{eq:Ham} , respectively.}
\subsection{Markov chains, mixing times, and canonical paths}\label{sec:markovchains}

Let $P$ be the transition matrix of a reversible Markov chain defined on a state space $\Omega$ with stationary distribution $\pi$.  Let $\mathbf{z}$ be the state of the chain at time $t = 0$, and define $P^t(\mathbf{z},\cdot)$ to be the distribution of the state of the chain  at time t.  The variation distance of the chain at time $t$ starting from the initial state $\mathbf{z}$ at time $t = 0$ is
\begin{equation}
d_{\mathbf{z}}(t) := \max_{A \subseteq \Omega} |P^t(\mathbf{z},A) - \pi(A)| = \frac{1}{2} \sum_{\mathbf{z'}\in \Omega}|P^t(\mathbf{z},\mathbf{z'})-\pi(\mathbf{z'})|
\end{equation}
Define $\tau(\epsilon)$ to be the worst-case time needed to be within variation distance $\epsilon$ of the stationary distribution,
\begin{equation}
\tau(\epsilon) := \max_{\mathbf{z}\in \Omega} \min_t \{t: d_{\mathbf{z}}(t') \leq \epsilon \; \; \forall t \geq t'\}.
\end{equation}
Define the mixing time to be $\tau_{\textrm{mix}} := \tau(1/4)$, and note that $\tau(\epsilon) = \tau_{\textrm{mix}} \log \epsilon^{-1}$~\cite{peres-2008}.  

To bound the mixing time of $P$ we use the method of canonical paths~\cite{jerrum-1996}.  A path from $x\in \Omega$ to $y \in \Omega$ is a sequence $\gamma_{x,y} = (v_1,\ldots,v_k)$ of states with $v_1 = x, v_k = y$ and $P(v_i,v_i+1) > 0$ for $i = 1,\ldots,k-1$.  A set $\mathcal{P} = \{\gamma_{x,y}\}$ containing a path that connects every pair of states is called a set of canonical paths.  If $\ell = \max_{\gamma \in \mathcal{P}} |\gamma|$ is the maximum length of any path in $\mathcal{P}$ and $\pi_{\min} := \min_{\mathbf{z} \in \Omega} \pi(\mathbf{z})$, then the mixing time $\tau_{\textnormal{mix}}$ is $\mathcal{O}(R \; \ell \ln \pi_{\min}^{-1})$, where the congestion $R$ is defined as  

\be
R := \max_{\left(v,v'\right) \in E(\Omega)} \;  \frac{1}{\pi(v)P(v,v')} \sum_{\substack{\gamma_{x,y}\in \mathcal{P} \\ (v,v') \in \gamma_{x,y}}} \pi(x)\pi(y)\label{eq:congestionDef}
\ee

In order to compute the congestion we will use a standard technique called encoding.  For an edge $e = \left(v,v'\right)$, let $\mathcal{P}(e)$ be the set of paths in $\mathcal{P}$ which pass through $e$.  An encoding is an assignment of an injective function $\eta_e:\mathcal{P}(e)\rightarrow \Omega$ to every edge $e \in E(\Omega)$.  More generally, it can be useful to consider injective encodings $\eta_e:\mathcal{P}(e)\rightarrow \Omega \times \Theta$ that map at most $G = |\Theta|$ paths through the $e$ to each state in $\Omega$~\cite{jerrum-1993}.  Let $\eta_e(\gamma_{x,y}) = (\eta_e(x,y),\theta_e(x,y))$, and suppose there is some $M$ such that
\begin{equation}
\pi(x)\pi\left(y\right) \leq M \pi\left(v\right) \pi\left(\eta_e\left(x,y\right)\right)  \; \; \; \forall x,y \in \Omega, \label{eq:simweight}
\end{equation}
then the congestion satisfies
\begin{eqnarray}
R &=& \label{eq:rho1} \max_{\left(v,v'\right) \in E(\Omega)} \;  \frac{1}{\pi(v)P(v,v')} \sum_{\substack{\gamma_{x,y}\in \mathcal{P} \\ (v,v') \in \gamma_{x,y}}} \pi(x)\pi(y)\\
&\leq&  \label{eq:applysimweight} \max_{\left(v,v'\right) \in E(\Omega)} \;  \frac{M\cdot G}{\pi(v)P(v,v')} \sum_{\substack{\gamma_{x,y}\in \mathcal{P} \\ (v,v') \in \gamma_{x,y}}} \pi(v) \pi(\eta_e(x,y)) \\
&\leq& \label{eq:inject}  M\cdot G \cdot P_{\min}^{-1}
\end{eqnarray}
where $P_{\min}:= \min_{(\mathbf{z},\mathbf{z}') \in E(\Omega)} P(\mathbf{z},\mathbf{z}')$, \eqref{eq:applysimweight} follows from property \eqref{eq:simweight}, and \eqref{eq:inject} follows from the injectivity property of the encoding.

\subsection{PIMC transition probabilities} \label{sec:PIMC-TP}
To sample from the classical Gibbs distribution \eqref{eq:generalQtCpartitionFunction} the PIMC method uses a Markov chain which chooses a site in the classical lattice $\Lambda$ uniformly at random and proposes to flip the bit at that site to make a transition $\mathbf{z} \rightarrow \mathbf{z'}$ with an acceptance probability $P(\mathbf{z},\mathbf{z'})$ given by the Metropolis rule~\cite{peres-2008},
\begin{equation}
P(\mathbf{z},\mathbf{z'}) :=\begin{cases} 
      \frac{1}{2 n L}\;\min \left \{1, \frac{\pi(\mathbf{z'})}{\pi(\mathbf{z})}\right\}, & \mathbf{z} \neq \mathbf{z'}\\
       1 - \sum_{\mathbf{z''} \neq z} P(\mathbf{z},\mathbf{z''}), & \mathbf{z} = \mathbf{z'}
   \end{cases} \label{metropolis}.
\end{equation}
As long as the state space is a connected graph, the transition matrix $P$ with transitions satisfying \eqref{metropolis} will converge to the stationary distribution $\pi$.

For the sake of lower bounding the minimum transition probability $P_{\textrm{min}}$ from the previous section it therefore suffices to lower bound $\pi(\mathbf{z'})/\pi(\mathbf{z})$ where $\mathbf{z}$ and $\mathbf{z'}$ differ only on a single site.  Suppose $\mathbf{z}$ differs from $\mathbf{z'}$ at position $(i,j)$, then after canceling common factors in \eqref{eq:generalQtCpartitionFunction} we have
\begin{equation}
\frac{\pi(\mathbf{z'})}{\pi(\mathbf{z})} = e^{\frac{1}{L}\L[A(z'_i) - A(z_i)\R]}  \frac{ \langle z_{i-1} | e^{\frac{B_{i-1}}{L}} | z'_i \rangle \langle z'_i | e^{\frac{B_i}{L}} | z_{i+1} \rangle}{ \langle z_{i-1} |e^{\frac{B_{i-1}}{L}} | z_i \rangle \langle z_i | e^{\frac{B_i}{L}} | z_{i+1} \rangle}.\label{eq:ratioAC}
\end{equation}
Depending on the off-diagonal matrix elements in $B_i$, the numerator of \eq{ratioAC} can potentially be zero.  This means that the state space $\Omega$ may not be connected by single bit flips, because the stationary distribution does not have support on all of $\Omega$.  This can be avoided in general by adding a 1-local transverse field field with sufficiently small magnitude $\Gamma$ to avoid disturbing the resulting approximation.   With this we can lower bound the ratio \eq{ratioAC} using the fact that $B_i$ is a nonnegative matrix, which implies
\be
\langle z | \left (I + \frac{B_i}{L}\right) | z'\rangle \leq \langle z | e^{\frac{B_i}{L}} | z'\rangle \quad \textrm{for all} \quad z,z' \in \{-1,1\}^n. \label{eq:elementnonneg}
\ee
The denominator of \eqref{eq:ratioAC} is upper bounded by 1, and \eq{elementnonneg} implies that the numerator is at least $(\beta \Gamma / L)^2$, therefore
\begin{equation}
\frac{\pi(\mathbf{z'})}{\pi(\mathbf{z})}  \geq  e^{-\frac{4\beta}{L}} \L(\frac{\beta \Gamma}{L}\R)^2 . \label{eq:transitionLB}
\end{equation}
\iffalse
Another important fact that we will repeatedly use is an upper bound on the ratio $\pi(\mathbf{z'})/\pi(\mathbf{z})$ when the configuration $\mathbf{z'}$ is equal to $\mathbf{z}$ everywhere except for 1 or 2 sites, and this difference causes $\mathbf{z'}$ to contain one additional 1-local or 2-local jump.  Adding a 2-local jump means that for some $i,j$ we have $z_{i,j} = z_{i+1,j}$ and $z_{i,j+1} = z_{i+1,j+1}$, and $z'_{i,j} \neq z'_{i+1,j}$ and $z'_{i,j+1} \neq z'_{i+1,j+1}$.  First note that
\be
\langle z_i | e^{\frac{B_i}{L}}|z_{i+1}\rangle = \prod_{j=1}^n \langle z_{i,j},z_{i,j+1} | e^{\frac{B_{i,j}}{L}}|z_{i+1,j},z_{i+1,j+1} \rangle, \label{eq:commuteweights}
\ee
and similarly for $\mathbf{z'}$. Therefore in the scenario we have described 
\be
\frac{\pi(\mathbf{z'})}{\pi(\mathbf{z})} \leq e^{\frac{4 \beta}{L}} \frac{\langle z'_{i,j},z'_{i,j+1} | e^{\frac{B_{i,j}}{L}}|z'_{i+1,j},z'_{i+1,j+1} \rangle}{\langle z_{i,j},z_{i,j+1} | e^{\frac{B_{i,j}}{L}}|z_{i+1,j},z_{i+1,j+1} \rangle}.
\ee
Now since $B_{i,j}$ is nonnegative and $\|B_{i,j}\| \leq 1$ we have
\be
\langle z'_{i,j},z'_{i,j+1} | e^{\frac{B_{i,j}}{L}}|z'_{i+1,j},z'_{i+1,j+1} \rangle  \leq \langle z'_{i,j},z'_{i,j+1} | I + \frac{2 B_{i,j}}{L}|z'_{i+1,j},z'_{i+1,j+1} \rangle
\ee
and finally since $z'_{i,j} \neq z'_{i+1,j}$ or $z'_{i,j} \neq z'_{i+1,j}$ we have that the decrease in the stationary probability due to adding a jump to a configuration is
\be
\frac{\pi(\mathbf{z'})}{\pi(\mathbf{z})} \leq e^{\frac{2 \beta}{L}}\left(\frac{2 \beta}{L}\right).
\ee
\fi
\subsection{Transverse Ising models}\label{sec:transverseIsing}
In this section we will describe a set of canonical paths which will be used to show rapid mixing of PIMC applied to Hamiltonians of the form \eqref{eq:tim}.  Let $\mathbf{x},\mathbf{y} \in \Omega$ be spin configurations with $\mathbf{x} = (x_{1,1},\ldots,x_{L,1},\ldots,x_{1,n}\ldots,x_{L,n})$ and $\mathbf{y} = (y_{1,1},\ldots,y_{L,1},\ldots,y_{1,n}\ldots,y_{L,n})$.  

It will be convenient to choose an ordering on the space $[L]\times [n]$.  We order bits first by the second coordinate and then the first, so that $(r,l)\preceq (s,m)$ if $l\leq m$ or if $r\leq s$ and $l=m$.
For any $1\leq r \leq L$ and $1 \leq m \leq n$ define
\begin{eqnarray}
[\leq (r,m)] &:=&  \left(\{1,\ldots,L\}\times \{1,\ldots,m-1\}\right) \cup \left(\{1,\ldots,r\}\times \{m\}\right)\\
{} [>(r,m)] &:=&   \left(\{r+1,\ldots,L\}\times \{m\}\right) \cup \left(\{1,\ldots,L\}\times \{m+1,\ldots,n\}\right).
\end{eqnarray}

The canonical path $\gamma_{\mathbf{x},\mathbf{y}}$ from $\mathbf{x}$ to $\mathbf{y}$ consists of simply replacing each bit of $\mathbf{x}$ with the corresponding bit of $\mathbf{y}$ following the above ordering.  This can be thought of as  $n$ steps, each of which consists of $L$ substeps that correspond to transitions of the Markov chain.  
The $r$-th substep of the $m$-th step consists of flipping the bit in the $r$-th Trotter slice of the $m$-th qubit if $x_{r,m}\neq y_{r,m}$ or leaving it alone if $x_{r,m}=y_{r,m}$.  In the former case, this contributes an edge $(\mathbf{v},\mathbf{v'})\in E(\Omega)$ with 
\begin{eqnarray}
\mathbf{v} \;&=&(y_{1,1},\ldots,y_{r - 1,m} , x_{r, m},x_{r+1,m},\ldots,x_{L,n})\\
\mathbf{v'}&=& (y_{1,1},\ldots,y_{r-1,m} , y_{r,m},x_{r+1,m}\ldots,x_{L,n}).
\end{eqnarray}
In the latter case (i.e.~$x_{r,m}=y_{r,m}$), this step does not add an edge to the path.  Thus the number of edges in $\gamma_{\mathbf{x},\mathbf{y}}$ equals the number of positions on which $\mathbf{x}$ and $\mathbf{y}$ disagree, which is always $\leq nL$.
In \secref{markovchains} we defined an {\em encoding} of a path to be an injective map $\eta_e$ from the set of paths through $e$ to the state space $\Omega$.  The encoding we'll use for this path is 
\begin{equation}
\eta_{(\mathbf{v},\mathbf{v'})}(\gamma_{\mathbf{x},\mathbf{y}})  := \mathbf{x}\cdot\mathbf{y}\cdot\mathbf{v},
%(x_{1,1},\ldots,x_{L,1},\ldots,x_{L - r-1,m} ,y_{L - r,m},y_{L-r+1,m}\ldots,y_{1,n},\ldots,y_{L,n})
\end{equation}
where  $\cdot$ to denote elementwise multiplication.
To see that $\mathbf{x},\mathbf{y}$ are uniquely determined by 
$\mathbf{v}$ together with $\boldsymbol{\eta} =\eta_{ (\mathbf{v},\mathbf{v'})}(\gamma_{\mathbf{x},\mathbf{y}})$, observe that $\mathbf{v}\cdot \boldsymbol{\eta} = \mathbf{x}\cdot\mathbf{y}$, which specifies the bits that are flipped along the path $\gamma_{\bx,\by}$. Moreover, the edge $(\bv,\bv')$ specifies the location $(r,m)$ of the bit being flipped in the transition from $\bv$ to $\bv'$.  
Let $(\bv\cdot\bfeta)_{\leq(r,m)}$ denote the vector which equals $\bv\cdot\bfeta=\bx\cdot\by$ on the coordinates in $[\leq (r,m)]$ and equals 1 on the coordinates in $[>(r,m)]$; define $(\bv\cdot\bfeta)_{>(r,m)}$ analogously.  Then we can recover $\bx$ by calculating $\bv'\cdot (\bv\cdot\bfeta)_{\leq(r,m)}$, and can recover $\by$ by calculating $\bv'\cdot  (\bv\cdot\bfeta)_{>(r,m)}$.

For any subset of the lattice sites $A \subseteq \Lambda$, define the restriction $E_\beta^A$ of the classical energy function $E_\beta$ by restricting the sum in \eq{timenergy} to sites $(i,j) \in A$, so that for all $\mathbf{z} \in \Omega$ we have
\be
E_\beta^A(\mathbf{z}) = \sum_{(i,j)\in A} \left( \frac{K^z_j}{L} z_{i,j} - \beta^{-1}J_i z_{i,j}z_{i+1,j} + \sum_{k \in A}\frac{K^{zz}_{jk}}{L} z_{i,j} z_{i,k} \right).
\ee

Now we compute
\ba
E_\beta(\mathbf{v}) &= E^{\leq(r,m)}_\beta(\mathbf{y})  + E^{>(r,m)}_\beta(\mathbf{x}) + B\\
E_\beta(\boldsymbol{\eta}) &= E^{\leq(r,m)}_\beta(\mathbf{x})  + E^{>(r,m)}_\beta(\mathbf{y}) + B'
%\pi(\mathbf{v}) = \pi(\mathbf{y})_{[\colon r, \colon m]}\pi(\mathbf{x})_{[r\colon,m\colon]} \cdot e^{-B}\\
%\pi(\boldsymbol{\eta})= \pi(\mathbf{x})_{[\colon r, \colon m]}\pi(\mathbf{y})_{[r\colon,m\colon]} \cdot e^{-B'}
\ea
where 
\begin{align*}
B &:=\beta^{-1}J_m\left(y_{1,m} x_{L,m} +  y_{r,m} x_{r+1,m}\right) +\frac{1}{L}\sum_{i = 1}^r  K^z_m y_{i,m} +\frac{1}{L}\sum_{i = r}^L K^{z}_{m} x_{i,m} \\
&+\frac{1}{L}\sum_{i = 1}^r K^{zz}_{m,m+1} y_{i,m} x_{i,m+1} + 
\frac{1}{L}\sum_{i = r}^L  K^{zz}_{m-1,m} y_{i,m-1} x_{i,m} \\
&+ \frac{1}{L}\sum_{i = 1}^L\sum_{j = 1}^{m-1} \sum_{k = m}^{n} K^{zz}_{j k} y_j x_k  ,\\ 
B' &:= \beta^{-1}J_m\left(x_{1,m} x_{L,m} +  x_{r,m} x_{r+1,m}\right) +\frac{1}{L} \sum_{i = 1}^r K^z_m x_{i,m} + \frac{1}{L} \sum_{i = r}^L  K^z_m y_{i,m}\\
&+\frac{1}{L} \sum_{i = 1}^r K^{zz}_{m,m+1} x_{i,m} y_{i,m+1}  +\frac{1}{L} \sum_{i = r}^L K^{zz}_{m-1,m} x_{i,m-1} y_{i,m}  \\
&+ \frac{1}{L}\sum_{i = 1}^L\sum_{j = 1}^{m-1} \sum_{k = m}^{n} K^{zz}_{j k} x_j y_k.
\end{align*}
We also need
\ba
E_\beta(\mathbf{x}) &= E^{\leq(r,m)}_\beta(\mathbf{x})  + E^{>(r,m)}_\beta(\mathbf{x}) + B_X\\
E_\beta(\mathbf{y}) &= E^{\leq(r,m)}_\beta(\mathbf{y})  + E^{>(r,m)}_\beta(\mathbf{y}) + B_Y
%\pi(\mathbf{x}) = \pi(\mathbf{x})_{[\colon r, \colon m]}\pi(\mathbf{x})_{[r\colon,m\colon]} \cdot e^{-B_X}\\
%\pi(\mathbf{y}) = \pi(\mathbf{y})_{[\colon r, \colon m]}\pi(\mathbf{y})_{[r\colon,m\colon]} \cdot e^{-B_Y}
\ea
where 
\begin{align*}
B_X &:= \sum_{i=1}^L \left( \frac{ K^z_m}{L} x_{i,m} + \beta^{-1} J_m x_{i,m}x_{i+1,m} \right)\\
&+ \frac{1}{L} \sum_{i = 1}^r K^{zz}_{m,m+1} x_{i,m} x_{i,m+1} +  \frac{1}{L} \sum_{i = r}^L  K^{zz}_{m-1,m}x_{i,m-1} x_{i,m}\\
&+ \frac{1}{L}\sum_{i = 1}^L\sum_{j = 1}^{m-1} \sum_{k = m}^{n} K^{zz}_{j k} x_j x_k,
\end{align*}
and $B_Y$ is defined similarly.  The value of $M$ needed to satisfy \eqref{eq:simweight} for this encoding can now be determined by
\be
\frac{\pi(\mathbf{x})\pi(\mathbf{y})}{\pi(\mathbf{v})\pi(\boldsymbol{\eta})} = e^{-\beta\left(E_\beta(\mathbf{x}) + E_\beta(\mathbf{y}) - E_\beta(\mathbf{v}) - E_\beta(\boldsymbol{\eta})\right)} = e^{\beta \left(B + B' - B_X - B_Y \right)}
\ee
To upper bound the exponent in the worst-case define  $J := \max_{m = 1,\ldots,n} J_m$, then
$$
\beta(B + B' - B_X - B_Y) \leq 8 J + 4 \beta + 4 \beta \sum_{j = 1}^{n/2} \sum_{k = \frac{n}{2} + 1}^{n} K^{zz}_{j k}.
$$
To upper bound the sum we use the assumption $K^{zz}_{j k} \leq |j - k|^{-(2 + \xi)}$ and regard it as a lower Riemann sum for an integral,  
$$
 \sum_{j = 1}^{n/2} \sum_{k = \frac{n}{2} + 1}^{n} |j - k|^{-(2 + \xi)} \leq \int_{j = 1}^{j = \frac{n}{2}} \int_{k = \frac{n}{2} + 1}^{k = n} dj dk \; |j - k|^{-(2 + \xi)} = \frac{1 -2^{\xi +1} n^{-\xi }+(n-1)^{-\xi }}{\xi ^2+\xi },
$$
and since $\xi>0$ it suffices in \eq{inject} to take $M = \mathcal{O}\left(e^{8 J + 4 \beta}\right) $.  If the 1D quantum system has periodic boundary conditions connecting qubit 1 to qubit $n$, then the same analysis would show that $M = \mathcal{O}\left(e^{8(J + \beta)}\right)$ suffices for \eqref{eq:simweight}.  Therefore by \eqref{eq:rho1}-\eqref{eq:inject} the congestion $R$ satisfies
\[
R \leq M \cdot P^{-1}_{\min} = \mathcal{O}\left( e^{8(\beta + J)}\cdot 2 n L \cdot e^{2 \left(J + \frac{\beta}{L}\right)} \cdot  \max_{(\mathbf{z},\mathbf{z'})} \left(\frac{\pi(\mathbf{z'}}{\pi(\mathbf{z})} \right) \right)
\]
the factor of $2 n L$ comes from the chain being lazy and the choice of random single-site updates.  Note that $\max_{(\mathbf{z},\mathbf{z'})} \left(\pi(\mathbf{z'}/\pi(\mathbf{z}) \right)$ is $\cO(L^2 \Gamma^{-2})$ from \eqref{eq:transitionLB} and also $\ell := \max_{\mathbf{x}, \mathbf{y}} |\gamma_{\mathbf{x},\mathbf{y}}| = n L$ and $\ln \pi_{\min}^{-1} = \ln (Z\cdot e^{n L J + n \beta}) \leq n L(J + \beta + \ln 2)$, using $Z \leq 2^{n L}$.   In \secref{qtc} we show that $L = \poly(n,\delta_{\textrm{mult}}^{-1})$ suffices to approximate the partition function with the desired precision.   Therefore the mixing time bounds in \secref{markovchains} imply that the PIMC method generates samples as was assumed in \secref{estimatePartitionFunction}, which completes the proof of theorem~\ref{theo:tim} for models of the form \eq{tim}.

\subsection{Polynomial time mixing in theorem 1} \label{sec:XX}

 In this section we will prove rapid mixing of PIMC for Hamiltonians of the form \eqref{eq:xxham} by describing paths between arbitrary states $\mathbf{x}, \mathbf{y} \in \Omega$.  It will be convenient to represent states in the form $\mathbf{z} = [\bar{z}_1,\ldots,\bar{z}_n]$ with $\bar{z}_j := (z_{1,j},\ldots,z_{L,j})$ denoting the spin values along the $j$-th worldline in configuration $\mathbf{z} \in \Omega$.  We say that a ``double jump'' occurs in configuration $\mathbf{z}$ at position $(i,j)$ if $\left(1+z_{i,j}z_{i+1,j}\right)\left(1+z_{i,j+1}z_{i+1,j+1} \right) \neq 0$.  For any consecutive pair of worldlines $j, j+1$, define the number of double jumps $d(\bar{z}_{j},\bar{z}_{j+1})$,
\be
d(\bar{z}_{j},\bar{z}_{j+1}):=\frac{1}{4}\sum_{i = 1}^L \left(1+z_{i,j}z_{i+1,j}\right)\left(1+z_{i,j+1}z_{i+1,j+1} \right).
\ee 
For any configuration $\mathbf{z}$ and worldline $j$ we will define a path segment to a new configuration $[\bar{z}_1,\ldots,\bar{z}_{j-2},\bar{z}'_{j-1},\bar{z}''_j,\bar{z}'_{j+1},\bar{z}_{j+2},\ldots,\bar{z}_n]$ with all of the double jumps removed from worldline $j$,
\be
d(\bar{z}'_{j-1},\bar{z}''_{j}) = d(\bar{z}''_{j},\bar{z}'_{j+1}) = 0 \label{eq:G2}
\ee
The double prime notation indicates that the $j$-th worldline does not contain any double jumps, while the single prime indicates a worldline that shares no jumps with its double primed neighboring worldline.   

There are many possible paths to choose from for removing double jumps in the $i$-th worldline.  Due to the form of the Hamitonian \eqref{eq:xxham}, changing the placement of a double jump along the imaginary-time direction has no effect on the energy of the effective classical system, except possibly for the part which arises from the diagonal part of the Hamiltonian.  Similar reasoning holds for the ordering of the single and double jumps; changing the ordering can only change the effective classical energy through the diagonal part of the Hamiltonian.  Similarly, changing the imaginary-time position of the double jumps in $\bar{x}_{i-1},\bar{x}_i$ will have no effect on the contributions to the effective classical energy that comes from double jumps occuring in $\bar{x}_{i-2},\bar{x}_{i-1}$.  Finally, note that changing the position of a double jump along the imaginary-time direction can only affect the classical energy through the diagonal part of the Hamiltonian, reducing the probability of the configuration by at most $e^{-4\beta}$ (the factor of 4 is due to the potential change from the diagonal part of the Hamiltonian when 3 worldlines are changed arbitrarily). 

From the above considerations it follows that the imaginary-time position of the double jumps in a particular worldline can be changed freely with the stationary weight of the resulting configuration decreased by at most a factor of $e^{-4\beta}$ (the factor of 4 bounds the potential change from the diagonal part of the Hamiltonian when 3 worldlines are changed arbitrarily).  Furthermore, since double jumps arising from the term $\sigma^{x}_i\sigma^x_{i+1}$  are simply bit-flips, two such jumps will cancel  out (``annihilate'') if they coincide.  This allows for paths that remove these double jumps in pairs simply by moving them around, with the additional guarantee that this will not lower the stationary weight by more than $e^{-4\beta}$.   Therefore the path from $[\bar{z}_1,\ldots,\bar{z}_n]$ to $[\bar{z}_1,\ldots,\bar{z}_{j-2},\bar{z}'_{j-1},\bar{z}''_j,\bar{z}'_{j+1},\bar{z}_{j+2},\ldots,\bar{z}_n]$ is defined to remove double jumps in pairs by shifting the position of the double jump at the least imaginary-time position along the imaginary-time direction until it encounters the next double jump and the pair is annihilated. 

The canonical path from $\mathbf{x}$ to $\mathbf{y}$ proceeds through the worldlines in order $1,\ldots,n$, prepping each pair of consecutive worldlines to remove double jumps, then updating the $j$-th worldline from $\bar{x}''_j$ to $\bar{y}''_j$.  The full update of the first worldline proceeds as,
\begin{align*}
[\bar{x}_1,\bar{x}_2,\ldots,\bar{x}_n] &\rightarrow  [\bar{x}''_1,\bar{x}'_2,\bar{x}_3,\ldots,\bar{x}_n] \rightarrow  [\bar{y}''_1,\bar{x}'_2,\bar{x}_3,\ldots,\bar{x}_n]\rightarrow [\bar{y}''_1,\bar{x}''_2,\bar{x}'_3,\bar{x}_4,\ldots,\bar{x}_n]\\&\rightarrow [\bar{y}''_1,\bar{y}''_2,\bar{x}'_3,\bar{x}_4,\ldots,\bar{x}_n] \rightarrow [\bar{y}_1,\bar{y}'_2,\bar{x}'_3,\bar{x}_4,\ldots,\bar{x}_n].
\end{align*}
The purpose of this series of steps (each of which consists of many single-site updates) is to avoid passing through intermediate configurations with too many additional double jumps occuring across the 1D domain wall formed near the worldlines which are being updated.  The worst case intermediate configurations along the paths will have the form,
\begin{equation}
\mathbf{z} = [\bar{y}_1,\ldots,\bar{y}_{j-2},\bar{y}'_{j-1}, (y''_{1,j},\ldots,y''_{r,j},x''_{r+1,j},\ldots,x''_{L,j}),\bar{x}'_{j+1}\bar{x}_{j+2},\ldots,\bar{x}_n].\label{eq:zee}
\end{equation}
This system of paths is highly degenerate (a particular configuration such as \eqref{eq:zee} will appear in many paths).   To bound the congestion we will use an encoding function $\eta_{(\mathbf{z},\mathbf{z}')} : \cP(\mathbf{z},\mathbf{z}')\rightarrow \Omega \times \Theta$.  If $\eta_{(\mathbf{z},\mathbf{z}')}(\gamma_{\mathbf{x}\mathbf{y}}) = (\mathbf{\eta},\theta)$, then
\begin{equation}
\mathbf{\eta}=  [\bar{x}_1,\ldots,\bar{x}_{j-2},\bar{x}'_{j-1}, (x''_{1,j},\ldots,x''_{r,j},y''_{r+1,j},\ldots,y''_{L,j}),\bar{y}'_{j+1},\bar{y}_{j+2},\ldots,\bar{y}_n].\label{eq:eta}
\end{equation}
and $\theta$ contains all of the information needed to reconstruct $\bar{x}_{j-1},\bar{x}_j,\bar{x}_{j+1},\bar{y}_{j-1},\bar{y}_j,\bar{y}_{j+1}$ from $\bar{x}'_{j-1},\bar{x}''_j$, $\bar{x}'_{j+1},\bar{y}'_{j-1},\bar{y}''_j,\bar{y}'_{j+1}$.  This lost information must specify the location of all of the jumps that have been removed from these 6 worldlines.  This means that the encoding sends exponentially paths to each state, and so to obtain a polynomial time mixing bound we have to slightly generalize \eq{inject} to account for the fact that the stationary probability on the endpoints of most of these paths is sufficiently small.  

For any $\mathbf{v} \in \Omega$ define $d_j(\mathbf{v}) :=d(\bar{v}_{j-1},\bar{v}_{j}) + d(\bar{v}_j,\bar{v}_{j+1})$ (which is just the total number of double jumps in the $j$-th worldline of $\mathbf{v}$).   Let $q$ be the number of additional double jumps in the $j$-th worldlines of $\mathbf{x}$ and $\mathbf{y}$ which are not present in the $j$-th worldlines of $\mathbf{z}$ and $\boldsymbol{\eta}$,
\begin{equation}
q =  d_j(\mathbf{x}) + d_j(\mathbf{y}) - d_j(\mathbf{z}) - d_j(\boldsymbol{\eta})\label{eq:que}
\end{equation}
Define $\mathcal{P}^q_{\mathbf{z},\mathbf{z}'}$ to be the subset of canonical paths passing through $(\mathbf{z},\mathbf{z}')$ with endpoints $\mathbf{x}$ and $\mathbf{y}$ that satisfy \eqref{eq:que}.   With these definitions the congestion \eqref{eq:congestionDef} becomes
\begin{equation}
R \leq \max_{(\mathbf{z},\mathbf{z}')} \frac{1}{\pi(\mathbf{z})P(\mathbf{z},\mathbf{z}')} \sum_{q=0}^{L} \sum_{\gamma_{\mathbf{x}\mathbf{y}} \in \mathcal{P}^q_{\mathbf{z},\mathbf{z}'}} \pi(\mathbf{x}) \pi(\mathbf{y}),\label{eq:rpartway}
\end{equation}
Let $\Theta = \cup_{q =1}^{6L} \Theta_q$ where $G(q) = |\Theta_q| = {{6L}\choose{q}}$ is the number of ways to restore $q$ jumps to the undetermined worldlines.  Now suppose we find $M(q)$ such that
\begin{equation}
\pi(\mathbf{x})\pi(\mathbf{y}) \leq M(q) \cdot \pi(\mathbf{z}) \pi(\boldsymbol{\eta}) \label{eq:Rintermediate},
\end{equation}
for all $\gamma_{\mathbf{x}\mathbf{y}} \in \mathcal{P}^q_{\mathbf{z},\mathbf{z}'}$, then this implies
\begin{align}
R & \leq P^{-1}_{\min}\sum_{q=0}^{L} M(q) \sum_{\gamma_{\mathbf{x}\mathbf{y}} \in \mathcal{P}^q_{\mathbf{z},\mathbf{z}'}} \pi(\boldsymbol{\eta})\\
&\leq P^{-1}_{\min}  \sum_{q=0}^{L} M(q) G(q) \sum_{\boldsymbol{\eta} \in \Omega}\pi(\boldsymbol{\eta})\\
& = P^{-1}_{\min}  \sum_{q=0}^{L} M(q) G(q) \label{eq:R3}
\end{align}
To satisfy \eqref{eq:Rintermediate} we need $M(q)$ to satisfy
\begin{align}
M(q) &\geq \frac{\pi(\mathbf{x})\pi(\mathbf{y})}{\pi(\mathbf{z})\pi({\boldsymbol{\eta}})} \label{eq:M1}\\
& =  \frac{\pi([\bar{x}_{j-2},\bar{x}_{j-1},\bar{x}_j,\bar{x}_{j+1},\bar{x}_{j+2}])\pi([\bar{y}_{j-2},\bar{y}_{j-1},\bar{y}_j,\bar{y}_{j+1},\bar{y}_{j+2}])}{\pi([\bar{y}_{j-2},\bar{y}'_{j-1},\bar{x}'_j + \bar{y}'_j,\bar{x}'_{j+1},\bar{x}_{j+2}])\pi([\bar{x}_{j-2},\bar{x}'_{j-1},\bar{x}'_j + \bar{y}'_j,\bar{y}'_{j+1},\bar{y}_{j+2}])}\label{eq:M2}\\
& \geq e^{4\beta}  \frac{\pi([\bar{x}_{j-1},\bar{x}_j,\bar{x}_{j+1}])\pi([\bar{y}_{j-1},\bar{y}_j,\bar{y}_{j+1}])}{\pi([\bar{y}'_j,\bar{x}'_j + \bar{y}'_j,\bar{x}'_{j+1}])\pi([\bar{x}'_{j-1},\bar{x}'_j + \bar{y}'_j,\bar{y}'_{j+1}])}\label{eq:M3}
\end{align}
where \eqref{eq:M2} to \eqref{eq:M3} follows because the the terms in the classical energy function which correspond to off-diagonal parts of the Hamiltonian involving qubits $j-2$ and $j+2$ are unaffected by the potentially changed spin values in the primed $j-1$ and $j+1$ worldlines.  The factor of $e^{4\beta}$ comes from upper bounding the change in the classical energy function due to the diagonal terms of the Hamiltonian between qubits $j-2,j-1$ and $j+1,j+2$.  The effect of the diagonal part of the Hamiltonian on the ratio expression remaining in \eqref{eq:M3} can be bounded by a factor of $e^{8 \beta}$.  The additional single jumps introduced near the domain wall can increase the ratio by at most $L^{12}$ (using \eqref{eq:transitionLB}, and noting that the worst case occurs when 2 new single jumps are introduced in each of 3 worldlines, in each of $\mathbf{z}$ and $\mathbf{\eta}$).   Each of the double jumps in the numerator which are not present in the denominator decreases the ratio by at least $(2\beta / L)$.  Therefore we may take $M(q) =  e^{12\beta}(2\beta)^q L^{-q}$ and \eqref{eq:R3} becomes
\be
R \leq e^{12\beta} L^{12} \cdot P^{-1}_{\min} \sum_{q=0}^{L} {{6 L}\choose{q}}  \left(\frac{2 \beta}{L}\right)^q \leq L^{12} \cdot e^{20\beta}\cdot P^{-1}_{\min}.
\ee
Using the fact that $P_{\min}^{-1}$ is $\poly(n,\beta,\Gamma^{-1})$ from \eqref{eq:transitionLB}, and the fact that the maximum length of the canonical paths used in this section is $\mathcal{O}(n^2 L^2)$, this shows that $\tau_{\textnormal{mix}}$ is $\poly(n,e^\beta,\Gamma^{-1})$ as required to complete the proof of theorem~\ref{theo:tim}.  
\subsection{Quasipolynomial time mixing in theorem 1}\label{sec:general1D}

In this section we also work with a restricted state space $\Omega^*$, which is defined as the set of configurations with fewer than $2 B := 4 \beta \log(n)$ jumps in each of the worldlines, 
\begin{equation}
\Omega^* = \{\mathbf{z} = [\bar{z}_1,\ldots,\bar{z}_n] \in \Omega : |\{i : z_{i,j} \neq z_{i+1,j}\}| \leq 2 B\;\textrm{for } j = 1,\ldots,n\}.
\end{equation}
In \secref{restricting} we show that $Z^* = \sum_{\mathbf{z}\in \Omega^*} e^{-\beta E_{\beta}(\mathbf{z})}$ satisfies $|Z - Z^*|\leq \delta Z$. The PIMC dynamics can be restricted to $\Omega^*$ by rejecting moves which would leave $\Omega^*$, or alternatively the canonical paths in this section bound the mixing time within $\Omega^*$ by the leaky random walk conductance bound in~\cite{crosson2016simulated}.  To show rapid mixing of PIMC within $\Omega^*$ we assign a canonical path $\gamma_{\mathbf{x y}}$ to every $\mathbf{x},\mathbf{y} \in \Omega^*$.  The simple case for these paths occurs when $\mathbf{x},\mathbf{y}$ are in a subset $\Omega^*_{\textrm{inner}}\subseteq \Omega^*$ which only allows half as many jumps per worldline as $\Omega^*$ does,
\begin{equation}
\Omega^*_{\textrm{inner}} = \{\mathbf{z} = [z_1,\ldots,z_n] \in \Omega : |\{i : z_{i,j} \neq z_{i+1,j}\}| \leq B \;\textrm{for } j = 1,\ldots,n\}.
\end{equation}  
If $\mathbf{x},\mathbf{y} \in \Omega^*_{\textrm{inner}}$ then we use the same paths which were used to show rapid mixing for TIM in Section~\ref{sec:transverseIsing}.  Various intermediate points $\mathbf{v}$ along this path $\gamma_{\mathbf{x},\mathbf{y}}$ may leave $\Omega^*_\textrm{inner}$ but will remain in $\Omega^*$ since
\begin{equation}
\mathbf{v} \;=(y_{1,1},\ldots,y_{L,1},\ldots,y_{L - r - 1,m} , x_{L - r, m},x_{L-r+1,m}\ldots,x_{1,n},\ldots,x_{L,n}),
\end{equation}
may have as many as $B$ jumps in the first $r$ time slices of $\bar{y}_m$ and $B$ jumps in the $L-r$ latter time slices of $\bar{x}_m$, so $\bar{v}_m$ will have at most $2B$ jumps in total and so $\mathbf{v}\in\Omega^*$.

For points $\mathbf{x},\mathbf{y} \in \Omega^*$ which are not in $\Omega^*_{\textrm{inner}}$, we must not produce intermediate points with too many jumps in any worldline.  To accomplish this the paths will have a ``clean up'' step similar to the one which was used in the previous section, to reduce the number of 2-local jumps in a worldline.   Consider an intermediate point of the path which has so far updated the worldlines of the first $m$ qubits (which could be $m = 0$), $\mathbf{v} = [\bar{y}_1,\ldots,\bar{y}_m,\bar{x}_{m+1},\ldots,\bar{x}_n]$.  If $\bar{x}_{m+1}$ has more than $B$ double jumps then we remove these in imaginary-time order.  Since we can no longer ``annihilate'' these transitions in pairs as was done in the previous section, we instead dismantle them into 1-local jumps (relying on the fictitious transverse field) and then annihilate those against one another.  The effect of this process on the stationary weight of the configuration will be analyzed after we have given the encoding function for these paths.  The net result of this process results is a worldline $\bar{x}'_{m+1}$ with fewer than $B$ jumps.  Similarly let $\bar{y}'_{m+1}$ be defined from $\bar{y}_{m+1}$ with all but the first $B$ jumps removed.   Next we update $\bar{x}'_{m+1}$ to $\bar{y}'_{m+1}$ in the usual way, and since $\bar{x}'_{m+1}$ and $\bar{y}'_{m+1}$ each have fewer than $B$ jumps, the intermediate configurations along this path will have less than $2B$ jumps in every worldline and therefore remain in $\Omega^*$.   

The encoding $\eta:\mathcal{P}_{(\mathbf{z},\mathbf{z}')} \rightarrow \Omega^*\times \Theta$ for these paths is again defined by \eqref{eq:eta}.  This time the maximum degeneracy of the encoding is $\binom{4 L}{B}$, since there are at most $4L$ positions to which at most $B$ 2-local jumps can be restored in the undetermined worldlines.  Unfortunately, this time we cannot have a inequality analogous to \eqref{eq:Rintermediate}, because the additional double jumps present in $\mathbf{x}$ and $\mathbf{y}$ can completely change the bit-values in wordlines $j-1$ and $j+1$.  Unlike the case in the previous section these altered spin values can reduce the classical energy drastically because of the $h$ terms acting on worldlines $j-2,j-1$ and $j+1,j+2$.   The stationary weight of $\mathbf{z}$ and $\boldsymbol{\eta}$ can be lowered from $\mathbf{x}$ and $\mathbf{y}$ by a factor which scales exponentially in the number of double jumps that $\mathbf{x}$ and $\mathbf{y}$ have in worldlines $j-2,j-1$ and $j+1,j+2$.   Therefore,
\begin{align}
M &\geq \frac{\pi(\mathbf{x})\pi(\mathbf{y})}{\pi(\mathbf{z})\pi({\boldsymbol{\eta}})} \label{eq:MMM1}\\
& =  \frac{\pi([\bar{x}_{i-2},\bar{x}_{i-1},\bar{x}_i,\bar{x}_{i+1},\bar{x}_{i+2}])\pi([\bar{y}_{i-2},\bar{y}_{i-1},\bar{y}_i,\bar{y}_{i+1},\bar{y}_{i+2}])}{\pi([\bar{y}_{i-2},\bar{y}'_{i-1},\bar{x}'_i + \bar{y}'_i,\bar{x}'_{i+1},\bar{x}_{i+2}])\pi([\bar{x}_{i-2},\bar{x}'_{i-1},\bar{x}'_i + \bar{y}'_i,\bar{y}'_{i+1},\bar{y}_{i+2}])}\label{eq:MMM2}\\
& \geq e^{4\beta} \cdot \left(\frac{\beta \Gamma}{L}\right)^{-4B}\frac{\pi([\bar{x}_{i-1},\bar{x}_i,\bar{x}_{i+1}])\pi([\bar{y}_{i-1},\bar{y}_i,\bar{y}_{i+1}])}{\pi([\bar{y}'_{i-1},\bar{x}'_i + \bar{y}'_i,\bar{x}'_{i+1}])\pi([\bar{x}'_{i-1},\bar{x}'_i + \bar{y}'_i,\bar{y}'_{i+1}])}\label{eq:MMM3}\\
& \geq e^{12\beta} \cdot \left(\frac{\beta \Gamma}{L}\right)^{-8B},
\end{align}
The fictitious transverse field has magnitude $\delta_\textrm{mult}/n$, so the congestion satisfies
\begin{align}
R &\leq  P^{-1}_{\min} M \cdot G\\
&\leq P_{\min}^{-1} \cdot e^{12\beta} \cdot \left(\textrm{poly}(n)\cdot \delta_\text{mult}^{-1}\right)^{8\beta\Gamma \log n} \cdot \binom{4 L}{4\beta \Gamma \log n},
\end{align}
from which it follows that the mixing time $\tau_{\textrm{mix}}$ is $\mathcal{O}(2^{\beta (\log n \delta_{\text{mult}}^{-1})^2})$ and we obtain the overall PIMC runtime as claimed in theorem~\ref{theo:tim}.

\section*{Acknowledgments}
Part of this work was completed while EC was funded by the Institute for Quantum Information and Matter, an NSF Physics Frontiers Center (NSF Grant PHY-1125565) with support of the Gordon and Betty Moore Foundation (GBMF-12500028). 
AWH was funded by NSF grants CCF-1452616, CCF-1629809, CCF-1729369 and ARO contract W911NF-17-1-0433.
The research is based upon work partially supported by the Office of
the Director of National Intelligence (ODNI), Intelligence Advanced
Research Projects Activity (IARPA), via the U.S. Army Research Office
contract W911NF-17-C-0050. The views and conclusions contained herein are
those of the authors and should not be interpreted as necessarily
representing the official policies or endorsements, either expressed or
implied, of the ODNI, IARPA, or the U.S. Government. The U.S. Government
is authorized to reproduce and distribute reprints for Governmental
purposes notwithstanding any copyright annotation thereon.

\end{document}